\pdfoutput=1
\RequirePackage{ifpdf}
\ifpdf 
\documentclass[pdftex]{sigma}
\else
\documentclass{sigma}
\fi

\numberwithin{equation}{section}

\newtheorem{Theorem}{Theorem}[section]
\newtheorem{Corollary}[Theorem]{Corollary}
\newtheorem{Lemma}[Theorem]{Lemma}
\newtheorem{Proposition}[Theorem]{Proposition}
\newtheorem{Conjecture}[Theorem]{Conjecture}
 { \theoremstyle{definition}
\newtheorem{Definition}[Theorem]{Definition}
\newtheorem{Remark}[Theorem]{Remark}
\newtheorem{Fact}[Theorem]{Fact}}

\newcommand{\setatop}[2]{\genfrac{}{}{0pt}{}{#1}{#2}}

\newcommand{\vlam}{\boldsymbol{\lambda}}

\renewcommand{\Re}{\mathrm{Re}}
\renewcommand{\Im}{\mathrm{Im}}

\newcommand{\bx}{\bar{x}}
\newcommand{\bs}{\bar{s}}

\newcommand{\ee}{{\rm e}}
\newcommand{\ii}{{\rm i}}

\newcommand{\R}{{\mathbb R}}
\newcommand{\C}{{\mathbb C}}
\newcommand{\Z}{{\mathbb Z}}

\newcommand{\cT}{{\mathcal T}}

\begin{document}

\allowdisplaybreaks

\newcommand{\arXivNumber}{2006.07171}

\renewcommand{\thefootnote}{}

\renewcommand{\PaperNumber}{105}

\FirstPageHeading

\ShortArticleName{Basic Properties of~Non-Stationary Ruijsenaars Functions}

\ArticleName{Basic Properties of~Non-Stationary Ruijsenaars\\ Functions\footnote{This paper is a~contribution to the Special Issue on Elliptic Integrable Systems, Special Functions and Quantum Field Theory. The full collection is available at \href{https://www.emis.de/journals/SIGMA/elliptic-integrable-systems.html}{https://www.emis.de/journals/SIGMA/elliptic-integrable-systems.html}}}

\Author{Edwin LANGMANN~$^\dag$, Masatoshi NOUMI~$^\ddag$ and Junichi SHIRAISHI~$^\S$}

\AuthorNameForHeading{E.~Langmann, M.~Noumi and J.~Shiraishi}

\Address{$^\dag$~Physics Department, KTH Royal Institute of~Technology, SE-106 91 Stockholm, Sweden}
\EmailD{\href{mailto:langmann@kth.se}{langmann@kth.se}}

\Address{$^\ddag$~Department of~Mathematics, KTH Royal Institute of~Technology, SE-100 44 Stockholm, Sweden\\
\hphantom{$^\ddag$}~(on leave from: Department of~Mathematics, Kobe University, Rokko, Kobe 657-8501, Japan)}
\EmailD{\href{mailto:noumi@kth.se}{noumi@kth.se}}

\Address{$^\S$~Graduate School of~Mathematical Sciences, The University of~Tokyo, Komaba,\\
\hphantom{$^\S$}~Tokyo 153-8914, Japan}
\EmailD{\href{mailto:shiraish@ms.u-tokyo.ac.jp}{shiraish@ms.u-tokyo.ac.jp}}

\ArticleDates{Received June 15, 2020, in final form October 08, 2020; Published online October 21, 2020}

\Abstract{For any variable number, a~non-stationary Ruijsenaars function was recently introduced as a~natural generalization of~an explicitly known asymptotically free solution of~the trigonometric Ruijsenaars model, and it~was conjectured that this non-stationary Ruijsenaars function provides an explicit solution of~the elliptic Ruijsenaars model. We~present alternative series representations of~the non-stationary Ruijsenaars functions, and we prove that these series converge. We~also introduce novel difference operators called~$\cT$ which, as we prove in the trigonometric limit and conjecture in the general case, act diagonally on the non-stationary Ruijsenaars functions.}

\Keywords{elliptic integrable systems; elliptic hypergeometric functions; Ruijsenaars systems}

\Classification{81Q80; 32A17; 33E20; 33E30}

\renewcommand{\thefootnote}{\arabic{footnote}}
\setcounter{footnote}{0}

\section{Introduction}
The celebrated quantum Calogero--Moser--Sutherland systems \cite{OP83} have natural relativistic gene\-ralizations discovered by Ruijsenaars \cite{R87}. The Ruijsenaars systems come in four kinds: ratio\-nal, trigonometric, hyperbolic, and elliptic, with the latter case being the most general and reducing to the others in certain limits~\cite{R87}. While the explicit solution of~the trigonometric Ruijsenaars model is known since a~long time: it~is given by the celebrated Macdonald polynomials~\cite{MacD}, and a~construction of~eigenfunctions of~the hyperbolic model was completed recently~\cite{HR}, only partial results about the explicit solution in the general elliptic case exist \cite{FV,R09A,R09B}. Recen\-tly, one of~us~(S) conjectured an explicit solution of~the elliptic Ruijsenaars model as a~limit of~special functions defined by explicit formal power series and called non-stationary Ruijsenaars functions \cite{S}. In particular, it~was shown in~\cite{S} that these functions reduce to the known solutions of~the trigonometric Ruijsenaars model in the trigonometric limit; they have several remarkable symmetry properties; and they arise in a~quantum field theory related to the elliptic Ruijsenaars system in a~way that is a~natural generalization of~how the known solutions of~the trigonometric Ruijsenaars model arise in a~quantum field theory related to the trigonometric Ruijsenaars model (this is only a~partial list of~results in~\cite{S}). The validity of~this conjecture was also tested by symbolic computer computations.

In this paper we prove some properties of~the non-stationary Ruijsenaars functions which, we~hope, will be useful to find proofs of~the conjectures in~\cite{S}.
In particular, we give alternative representations of~these functions which are simpler than the original definitions; we prove that the series defining these functions are absolutely convergent in a~suitable domain;
and we present novel difference operators, called~$\cT$, which, we conjecture, acts diagonally on the non-stationary Ruijsenaars functions (by this we mean that the latter are eigenfunctions of~the former).

\medskip

\noindent {\bf Notation:} Throughout the paper, the symbols $q$, $t$, $p$, $\kappa$ (complex parameters) and~$N$ (variable number) have special significance. We~use the following standard notation,
\begin{gather*}
(z;q)_\infty \equiv \prod_{n=0}^\infty\big(1-zq^n\big)\qquad (|q|<1),
\\
(z;q)_k\equiv \frac{(z;q)_\infty}{(q^k z;q)_\infty} \qquad (k\in\Z),
\\
(z;q,p)_\infty\equiv \prod_{n,m=0}^\infty \big(1-q^np^m z\big)\qquad (|q|<1,\ |p|<1),
\\
\theta(z;p) \equiv (z;p)_\infty(p/z;p)_\infty
\end{gather*}
for~$z\in\C$. Moreover, $T_{q,z}=q^{z\partial_z}$, i.e.,
\begin{gather*}
(T_{q,z}f)(z)=f(qz)
\end{gather*}
for functions $f(z)$ of~$z\in\C$.
For $z\in\C$, $\Re(z)$ and $\Im(z)$ are the real- and imaginary parts of~$z$, and $\sin\arg(z) = \Im(z)/|z|$. For $x=(x_1,\dots,x_N)$ and $\lambda=(\lambda_1,\dots,\lambda_N)$, $x^\lambda$ is short for~$x_1^{\lambda_1}\cdots x_N^{\lambda_N}$, $x^{-1}$ is short for~$\big(x_1^{-1},\dots,x_N^{-1}\big)$, and $x^{+1}=x$.
We~denote as $\C[[z_1,\dots,z_N]]$ the space of~all formal power series $f(z)=\sum_{\mu\in\Z_{\geq 0}^N}c_\mu z_1^{\mu_1}\cdots z_N^{\mu}$ in formal variables $z=(z_1,\dots, z_N)$ with complex coefficients~$c_\mu$.

\section{Prerequisites}
We~recall some known facts about the Macdonald polynomials \cite{MacD} and certain special functions generalizing the Macdonald polynomials and constructed so as to solve the trigonometric Ruijsenaars model \cite{NS,S05} (Section~\ref{sec:trig}). We~also recall the eigenvalue problem defining the elliptic Ruijsenaars model, and the definition of~the non-stationary Ruijsenaars functions (Section~\ref{sec:nonstat}).

\subsection{Trigonometric Ruijsenaars model}\label{sec:trig}
For fixed $N\in\Z_{\geq 1}$, the Macdonald polynomials $P_\lambda(x;q,t)=P_\lambda\big(x;q^{-1},t^{-1}\big)$ are symmetric polynomials in variables $x=(x_1,\dots,x_N)\in\C^N$ depending on two complex parameters $q$, $t$ and labeled by partitions $\lambda$ of~length less than or equal to $N$, i.e., $\lambda=(\lambda_1,\dots,\lambda_N)$ with $\lambda_i\in \Z_{\geq 0}$ such that $\lambda_1\geq \lambda_2\geq \dots \geq \lambda_N\geq 0$.
They can be defined as common eigenfunctions of~the following commuting {\em Macdonald--Ruijsenaars operators},
\begin{gather}\label{MR}
D^{\pm}_N(x|q,t) \equiv \sum_{i=1}^N \prod_{j\neq i}^N \frac{\big(1-t^{\pm 1}x_i/x_j\big)}{(1-x_i/x_j)} T_{q,x_i}^{\pm 1}
\end{gather}
with corresponding eigenvalues $\sum_{i=1}^N t^{\pm(N-i)}q^{\pm{\lambda_i}}$, together with a~convenient normalization condition~\cite{MacD}.

The operators $D^\pm_N(x|q,t)$ are related by similarity transformations to the operators defining the trigonometric Ruijsenaars model \cite{R87}.

As conjectured by one of~us (S) \cite{S05} and proved by two of~us (NS) \cite{NS}, these eigenfunctions are naturally generalized to a~special function $f_N(x|s|q,t)$ depending on another set of~variables, $s=(s_1,\dots,s_N)\in\C^N$, and determined by the following requirement, up to normalization:
for~$\lambda\in\C^N$, the function
\begin{gather}
\label{fs}
x^{\lambda} f_N(x|s|q,t),\qquad s_i=t^{N-i}q^{\lambda_i}
\end{gather}
is a~common eigenfunction of~$D^\pm_N(x|q,t)$ with corresponding eigenvalue $\sum_{j=1}^N s_j^{\pm 1}$; if $\lambda$ is a~partition, then the function in~\eqref{fs} is equal to the Macdonald polynomial $P_\lambda(x;q,t)$ \cite{NS}.
The function $f_N(x|s|q,t)$ is called the {\em asymptotically free solution of~the trigonometric Ruijsenaars model}.

One remarkable property of~this function is that it~has a~simple explicit series representation which converges absolutely in a~suitable domain \cite{NS}:\footnote{Note that $f_N(x|s|q,t)$ here is $p_N(x;s|q,t)$ in \cite{NS}.}
\begin{gather}
\label{fN}
f_N(x|s|q,t)= \sum_{\theta\in\mathsf{M}_N} c_{N}(\theta|s|q,t) \prod_{1\leq i<k\leq N}(x_k/x_i)^{\theta_{ik}}
\end{gather}
with $\mathsf{M}_N$ the set of~$N\times N$ strictly upper triangular matrices with nonnegative integer entries:
\begin{gather}
\label{MMN}
\mathsf{M}_N\equiv \big\{ \theta=(\theta_{ik})_{i,k=1}^N\,|\, \theta_{ik}\in\Z_{\geq 0}\ (\forall\, i,k), \ \theta_{ik}=0\ (1\leq k\leq i\leq N)\big\} ,
\end{gather}
and\footnote{We~write~\eqref{cN} in a~way that emphasizes the similarity with~\eqref{cNinfty} below, for reasons that will become clear later on. Due to this, we include the empty factors for~$i=N$.}
\begin{gather}
c_{N}(\theta|s|q,t) =
\prod_{i=1}^N \prod_{i< j\leq k\leq N}
\frac{\big(q^{\sum_{a> k}(\theta_{ia} - \theta_{ja})}ts_{j}/s_i;q\big)_{\theta_{ik}}}
{\big(q^{\sum_{a> k}(\theta_{ia} - \theta_{ja})}qs_{j}/s_i;q\big)_{\theta_{ik}}}\nonumber
\\ \hphantom{c_{N}(\theta|s|q,t) =}
{}\times
\prod_{i=1}^N \prod_{i\leq j<k\leq N}
\frac{\big(q^{-\theta_{jk}-\sum_{a>k}(\theta_{ja}-\theta_{ia})}qs_{j}/ts_i ;q\big)_{\theta_{ik}}}{\big(q^{-\theta_{jk}-\sum_{a>k}(\theta_{ja}-\theta_{ia})}s_{j}/s_i ;q\big)_{\theta_{ik}}}
\label{cN}
\end{gather}
(note that~\eqref{fN}--\eqref{cN} is equivalent to (1.10)--(1.11) in~\cite{NS}).

For later reference, we also define the function\footnote{Note that $\varphi_N(x|s|q,t)$ here is $\psi_N(x;s|q,t)$ in \cite{NS}.}
\begin{gather}
\label{phiN}
\varphi_N(x|s|q,t)\equiv \prod_{1\leq i<j\leq N} \frac{(qx_j/tx_i;q)_\infty}{(qx_j/x_i;q)_\infty} f_N(x|s|q,t),
\end{gather}
which, as proved in \cite{NS}, has the following remarkably symmetry properties:
\begin{alignat}{3}
&\varphi_N(x|s|q,t)= \varphi_N(s|x|q,t) \quad && \text{(bispectral duality)},& \nonumber\\
& \varphi_N(x|s|q,t)= \varphi_N(s|x|q,q/t) \quad && \text{(Poincar\'e duality)}.& \label{symmetry}
\end{alignat}

\subsection{Non-stationary Ruijsenaars functions}\label{sec:nonstat}

The analogue of~the operators in~\eqref{MR} for the elliptic Ruijsenaars model depends on a~further complex parameter, $p$ such that $|p|<1$:{\samepage
\begin{gather}\label{R}
D_N^\pm(x|q,t,p) \equiv \sum_{i=1}^N \prod_{j\neq i}^N \frac{\theta\big(t^{\pm 1}x_i/x_j;p\big)}{\theta(x_i/x_j;p)} T_{q,x_i}^{\pm 1}
\end{gather}
with the theta function $\theta(z;p)$ given in the introduction; note that $D_N^\pm(x|q,t)=D_N^\pm(x|q,t,0)$.}

The non-stationary Ruijsenaars function $f^{\widehat{\mathfrak{gl}}_N}(x,p|s,\kappa|q,t)$ is a~conjectured eigenfunction of~a~deformation of~the operators in~\eqref{R}, depending on a~further complex parameter, $\kappa$, and reducing to the operators in~\eqref{R} in the limit $\kappa\to 1$ \cite{S}.

\begin{Definition}[non-stationary Ruijsenaars functions]\label{def1}
For $N\in\Z_{\geq 1}$, four parameters $q$, $t$, $p$, $\kappa$, and two sets of~variables $x=(x_1,\dots,x_N)$ and $s=(s_1,\dots,s_N)$, the non-stationary Ruijsenaars function is defined as a~formal power series in $(px_2/x_1,\dots,px_{N}/x_{N-1},px_1/x_N)$ as follows,
\begin{gather}\label{f1}
f^{\widehat{\mathfrak{gl}}_N}(x,p|s,\kappa|q,t)\equiv \!\!\! \sum_{\lambda^{(1)},\dots,\lambda^{(N)}\in\mathsf{P}} \prod_{i,j=1}^N\frac{\mathsf{N}^{(j-i|N)}_{\lambda^{(i)},\lambda^{(j)}}(ts_j/s_i|q,\kappa)}
{\mathsf{N}^{(j-i|N)}_{\lambda^{(i)},\lambda^{(j)}}(s_j/s_i|q,\kappa)}
\prod_{\beta=1}^N\prod_{\alpha\geq 1}(p x_{\alpha+\beta}/tx_{\alpha+\beta-1})^{\lambda^{(\beta)}_\alpha}\!\!\!
\end{gather}
with $x_{\alpha+\ell N}\equiv x_{\alpha}$ for all $\alpha=1,\dots,N$ and $\ell\in\Z_{\geq 1}$, $\mathsf{P}$ the set of~all partitions $\lambda$ of~arbitrary length, i.e., $\lambda=(\lambda_1,\lambda_2,\dots)$ with $\lambda_i\in\Z_{\geq 0}$ such that $\lambda_1\geq \lambda_2\geq \cdots$ and $\lambda_i=0$ for~$i\gg 0$, and
\begin{gather}
\label{N1}
\mathsf{N}^{(k|N)}_{\lambda,\mu}(u|q,\kappa)\equiv \!\!\!\!\!\!\!\!
\prod_{\setatop{b\geq a\geq 1}{b-a\equiv k(\mathrm{mod}N)}}\!\!\!\!\!\!\!\!
\big(u q^{-\mu_a+\lambda_{b+1}}\kappa^{-a+b};q\big)_{\lambda_{b}-\lambda_{b+1}}
\!\!\!\!\!\!\!\!\!\!\!\!\!\!\!\prod_{\setatop{\beta\geq \alpha\geq 1}{\beta-\alpha\equiv (-k-1)(\mathrm{mod}N)}}\!\!\!\!\!\!\!\!\!\!\!\!\!
\big(uq^{\lambda_\alpha-\mu_\beta}\kappa^{\alpha-\beta-1};q\big)_{\mu_\beta-\mu_{\beta+1}}\!\!\!
\end{gather}
for~$\lambda$, $\mu\in\mathsf{P}$, $k\in\Z/N\Z$, and $u\in\C$.
\end{Definition}

As discussed in \cite{S}, the expressions in~\eqref{N1} are {\em Nekrasov factors} \cite{N}.
Moreover, by chan\-ging $(p,\kappa)\to\big(p^{1/N},\kappa^{1/N}\big)$ and scaling variables $x\to p^{\delta/N}x=\big(p^{(N-1)/N}x_1,p^{(N-2)/2}x_2,\dots,x_N\big)$ and~similarly for~$s$, one obtains a~function, $f^{\widehat{\mathfrak{gl}}_N}\big(p^{\delta/N}x,p^{1/N}|\kappa^{\delta/N} s,\kappa^{1/N}|q,q/t\big)$, that converges to the asymptotically free solution of~the Ruijsenaars model, $f_N(x|s|q,t)$, in the limit $p\to 0$~\cite{S}.

\begin{Remark}\label{rem:balanced}
To explain the scaling just mentioned, we point out one important technical point: in Definition~\ref{def1}, equations~\eqref{phiNinfty} and~\eqref{symmetry1} below, and equations~\eqref{cTell2} and~\eqref{cTBeq} in~Section~\ref{sec:cT}, we use {\em balanced coordinates} $x_B$, $p_B$, $s_B$, $\kappa_B$ and $t_B$ (written without the subscript~$B$ for~simplicity), whereas elsewhere in the paper we use {\em unbalanced coordinates} $x_U$, $p_U$, $s_U$, $\kappa_U$ and $t_U$ (also written without subscript $U$) related to the balanced coordinates as follows,
\begin{gather*}
(x_B)_i = (p_U)^{(N-i)/N}(x_U)_i \qquad (i=1,\dots,N), \qquad p_B =(p_U)^N,
\\
(s_B)_i = (\kappa_U)^{(N-i)/N}(s_U)_i \qquad (i=1,\dots,N), \qquad
\kappa_B=(\kappa_U)^N,\qquad t_B = q/t_U .
\end{gather*}
Thus, the scaling just described can be understood as a~transformation from balanced to unbalanced coordinates.
\end{Remark}

The main conjecture in \cite{S} is that eigenfunction of~the operator in~\eqref{R} can be obtai\-ned by dividing this rescaled function $f^{\widehat{\mathfrak{gl}}_N}\big(p^{\delta/N}x,p^{1/N}|\kappa^{\delta/N} s,\kappa^{1/N}|q,q/t\big)$ by a~(known) factor $\alpha\big(p^{1/N}|\kappa^{\delta/N}s,\kappa^{1/N}|q,t\big)$ and taking the limit $\kappa\to 1$; see Conjecture~1.14 in~\cite{S}. One important open problem is to find the operator depending on $\kappa$ having these rescaled non-stationary Ruijsenaars functions as eigenfunctions and reducing to the Macdonald--Ruijsenaars operator in~\eqref{R} in the limit $\kappa\to 1$.\footnote{There is, however, a~recent proposal mentioned in Section~\ref{sec:remarks}.} At this point, this operator is only known in limiting cases: the non-relativistic limit $q\to 1$ where the Ruijsenaars systems reduce to the non-stationary elliptic Calogero--Sutherland system \cite{S}, and the limit $t,p\to 0$ with fixed $p/t$ leading to the affine Toda system~\cite{S}. We~stress that the non-stationary $\cT$-operators introduced in this paper do {\em not} reduce to the elliptic Macdonald--Ruijsenaars operators in the limit $\kappa\to 1$: the $\cT$-operators are of~a different kind, and they are new even in the trigonometric limit; only the affine Toda limit of~the non-stationary $\cT$-operator was known before~\cite{S}.

A natural generalization of~the function in~\eqref{phiN} is
\begin{gather}
\varphi^{\widehat{\mathfrak{gl}}_N}(x,p|s,\kappa|q,t)\equiv
\prod_{1\leq i<j\leq N} \frac{\big(qp^{j-i}x_j/tx_i;q,p^N\big)_\infty}{\big(qp^{j-i}x_j/x_i;q,p^N\big)_\infty}\nonumber
\\ \hphantom{\varphi^{\widehat{\mathfrak{gl}}_N}(x,p|s,\kappa|q,t)\equiv}
{}\times\prod_{1\leq i\leq j\leq N} \frac{\big(qp^{N-j+i}x_i/tx_j;q,p^N\big)_\infty}{\big(qp^{N-j+i}x_i/x_j;q,p^N\big)_\infty}
 f^{\widehat{\mathfrak{gl}}_N}(x,p|s,\kappa|q,t)
\label{phiNinfty}
\end{gather}
and, as conjectured in \cite{S}, it~has the following symmetry properties generalizing
the ones in~\eqref{symmetry}.

\begin{Conjecture}The functions in~\eqref{phiNinfty} satisfy
\begin{alignat}{3}
& \varphi^{\widehat{\mathfrak{gl}}_N}(x,p|s,\kappa|q,t)= \varphi^{\widehat{\mathfrak{gl}}_N}(s,\kappa|x,p|q,t)\quad && \text{$($bispectral duality$)$},&\nonumber\\
&\varphi^{\widehat{\mathfrak{gl}}_N}(x,p|s,\kappa|q,t)= \varphi^{\widehat{\mathfrak{gl}}_N}(x,p|s,\kappa|q,q/t) \quad && \text{$($Poincar\'e duality$)$}.& \label{symmetry1}
\end{alignat}
\end{Conjecture}

\section{Results on the non-stationary Ruijsenaars function}\label{sec:results}

We~give alternative series representations of~the non-stationary Ruijsenaars functions (Section~\ref{sec:alternative}) and prove convergence of~these series in a~suitable domain (Section~\ref{sec:convergence}).
\subsection{Alternative series representations}
\label{sec:alternative}
Our first result makes manifest that the non-stationary Ruijsenaars function in~\eqref{f1}--\eqref{N1} is a~natural generalization of~the asymptotically free solutions of~the trigonometric Ruijsenaars model in~\eqref{fN}--\eqref{cN} .
For that, we extend the variables $x=(x_i)_{i=1}^N$ and $s=(s_i)_{i=1}^N$ to infinitely many variables $\bx=(x_i)_{i=1}^\infty$ and $\bs=(s_i)_{i=1}^\infty$; as we will see, the pertinent extension
is provided by the parameters $p$ and $\kappa$, respectively --- see~\eqref{period1}.

We~first introduce a~natural generalization of~the function in~\eqref{fN}--\eqref{cN} to infinitely many variables.
\begin{Definition}
\label{def2}
For $N\in\Z_{\geq 1}$, two parameters $q$, $t$, and two sets of~infinitely many variables $\bx=(x_1,x_2,\dots)$ and $\bs=(s_1,s_2,\dots)$, let the following define a~formal power series in the infinitely many variables
$(x_2/x_1,x_3/x_2,x_4/x_3,\dots)$,
\begin{gather}
\label{fNinfty}
 f_{N,\infty}(\bx|\bs|q,t) \equiv \sum_{\theta\in\hat{\mathsf{M}}_N} c_{N,\infty}(\theta|\bs|q,t) \prod_{i=1}^N \prod_{k>i} (x_k/x_i)^{\theta_{ik}}
\end{gather}
with $\hat{\mathsf{M}}_N$ the set of~infinite, $N$-periodic, strictly upper triangular matrices with nonnegative integer entries which are non-zero only in a~finite strip away from the diagonal:
\begin{gather}
\label{hMMN}
\hat{\mathsf{M}}_N\equiv \big\{ \theta=(\theta_{ik})_{i,k=1}^\infty\,|\, \theta_{ik}=\theta_{i+N,k+N}\in\Z_{\geq 0}\ (i,k\geq 1), \ \theta_{ik}=0\ (k\leq i,\ k\gg i)\big\} ,
\end{gather}
and
\begin{gather}
c_{N,\infty}(\theta|\bs|q,t)
\equiv
\prod_{i=1}^N \prod_{i< j\leq k< \infty}
\frac{\big(q^{\sum_{a> k}(\theta_{ia} - \theta_{ja})}ts_{j}/s_i;q\big)_{\theta_{ik}}}{\big(q^{\sum_{a> k}(\theta_{ia} - \theta_{ja})}qs_{j}/s_i;q\big)_{\theta_{ik}}}\nonumber
\\ \hphantom{c_{N,\infty}(\theta|\bs|q,t)\equiv }
{}\times\prod_{i=1}^N \prod_{i\leq j<k<\infty}
\frac{\big(q^{-\theta_{jk}-\sum_{a>k}(\theta_{ja}-\theta_{ia})}qs_{j}/ts_i ;q\big)_{\theta_{ik}}}{\big(q^{-\theta_{jk}-\sum_{a>k}(\theta_{jb}-\theta_{ia})}s_{j}/s_i ;q\big)_{\theta_{ik}}} .
\label{cNinfty}
\end{gather}

\end{Definition}

Note that the product in~\eqref{cNinfty} always contains only a~finite number of~factors different from $1$. Moreover, by the condition $\theta_{ik}=\theta_{i+N,k+N}$,
a matrix $\theta\in\hat{\mathsf{M}}_N$ is fully determined by the matrix elements $\theta_{ik}$ for~$1\leq i\leq N$ and $1\leq k<\infty$.
Furthermore, matrices in $\mathsf{M}_N$ can be naturally identified with matrices $\theta$ in $\hat{\mathsf{M}}_N$ by setting $\theta_{ik}=0$ if $i>N$, or $k>N$, or both.

To state out result we use the $N$-vector $\delta\equiv (\delta_1,\dots,\delta_N)$ with $\delta_i=N-i$, and the notation $p^{\delta/N}x$ and $\kappa^{\delta/N}s$ for the $N$-vectors with components
$\big(p^{\delta/N} x\big)_i=p^{(N-i)/N}x_i$ and $\big(\kappa^{\delta/N} s\big)_i=\kappa^{(N-i)/N}s_i$, respectively ($i=1,\dots,N$). As explained in Remark~\ref{rem:balanced}, this can be understood as a~transformation going from balanced to unbalanced coordinates.

\begin{Theorem}\label{thm1}
The non-stationary Ruijsenaars function in~\eqref{f1}--\eqref{N1} is related to the function in~\eqref{fNinfty}--\eqref{cNinfty} as follows,
\begin{gather}\label{Identity1}
f^{\widehat{\mathfrak{gl}}_N}\big(p^{\delta/N}x,p^{1/N}|\kappa^{\delta/N} s,\kappa^{1/N}|q,q/t\big) = f_{N,\infty}(\bx|\bs|q,t)
\end{gather}
with the variables $x=(x_i)_{i=1}^N$ and $s=(s_i)_{i=1}^N$ on the left-hand side extended to variables $\bx=(x_i)_{i=1}^\infty$ and $\bs=(s_i)_{i=1}^\infty$ on the right-hand side by the rules\footnote{``$x_{i+N}=px_i$ $(i\geq 1)$" is short for ``$x_{i+kN}=p^kx_i$ $(i=1,\dots,N,\; k\in\Z_{\geq 1})$".}
\begin{gather}
\label{period1}
x_{i+N}=px_i,\qquad s_{i+N}=\kappa s_i \qquad (i\geq 1).
\end{gather}
\end{Theorem}
(The proof is by straightforward computations given in Appendix~\ref{app:proof}.)

In the following, it~is sometimes convenient to use a~notation for the functions $f_{N,\infty}$ that emphasizes that the arguments $\bx$ and $\bs$ are fixed by $x$, $s$, $p$ and $\kappa$:

\begin{Definition}
\label{def:fNinfty}
We~write
\begin{gather*}
f_{N,\infty}(x,p|s,\kappa|q,t) \equiv f_{N,\infty}(\bx|\bs|q,t)
\end{gather*}
if $\bx\,{=}\,(x_1,x_2,\dots)$ and $\bs\,{=}\,(s_1,s_2,\dots)$ on the right-hand side are determined by $x\,{=}\,(x_1,\dots,x_N)$, $p$, $s=(s_1,\dots,s_N)$, and $\kappa$ as in~\eqref{period1}.
Thus
\begin{gather*}
f_{N,\infty}(x,p|s,\kappa|q,t) = \sum_{\theta\in\hat{\mathsf{M}}_N} c_{N,\infty}(\theta|s,\kappa|q,t)e_{N,\infty}(\theta|x,p)
\end{gather*}
with
\begin{gather}\label{cNeNdef}
c_{N,\infty}(\theta|s,\kappa|q,t) \equiv c_{N,\infty}(\theta|\bs|q,t),
\qquad
e_{N,\infty}(\theta|x,p)\equiv \prod_{i=1}^N \prod_{k=i+1}^\infty (x_k/x_i)^{\theta_{ik}}
\end{gather}
and the identifications in~\eqref{period1} on the right-hand side in~\eqref{cNeNdef}.
\end{Definition}

Theorem~\ref{thm1} makes manifest the following important result in \cite{S}: After suitably scaling the variables, the non-stationary Ruijsenaars function reduces the asymptotically free solution of~the trigonometric Ruijsenaars model, $f_N(x|s|q,t)$~\eqref{fN}--\eqref{cN}, in the limit $p\to 0$; in particular, it~becomes independent of~$\kappa$ in this limit:

\begin{Corollary}
We~have
\begin{gather}
\label{cor1}
\lim_{p\to 0} f^{\widehat{\mathfrak{gl}}_N}\big(p^{\delta/N}x,p^{1/N}|\kappa^{\delta/N} s,\kappa^{1/N}|q,q/t\big) = f_{N}(x|s|q,t).
\end{gather}
\end{Corollary}
\begin{proof}
By Theorem~\ref{thm1},~\eqref{cor1} is equivalent to
\begin{gather*}
\lim_{p\to 0} f_{N,\infty}(x,p|s,\kappa|q,t) =f_N(x|s|q,t),
\end{gather*}
but this is obvious from definitions: by~\eqref{period1}, $(x_k/x_i)\to 0$ for~$k>N$ as $p\to 0$; therefore, the sum over $\theta\in\hat{\mathsf{M}}_N$ on the right-hand side in~\eqref{fNinfty} collapses to a~sum over $\theta\in \mathsf{M}_N$ in this limit; obviously, for~$\theta\in\mathsf{M}_N$, the coefficients $c_{N,\infty}(\theta|\bs|q,t)$ in~\eqref{cNinfty} do not depend on $s_{i>N}$ and are identical with the coefficients $c_N(\theta|s|q,t)$ in~\eqref{cN}.
\end{proof}

We~prove Theorem~\ref{thm1} by a~direct computation in Appendix~\ref{app:proof}. This proof uses an alternative representation of~the function $f_{N,\infty}(\bx|\bs|q,t)$ which is interesting in its own right:

\begin{Lemma}
\label{lem1}
The formal power series in~\eqref{fNinfty}--\eqref{cNinfty} can be written as
\begin{gather}
\label{fNinfty1}
 f_N(\bx|\bs|q,t) = \sum_{\vlam\in\mathsf{P}^N} C_{N,\infty}(\vlam|\bs|q,t) \prod_{i=1}^N\prod_{k\geq 1}(x_{i+k}/x_{i+k-1})^{\lambda^{(i)}_k},
\end{gather}
with $\mathsf{P}^N$ the set of~all $N$-partitions $\vlam=\big(\lambda^{(1)},\lambda^{(2)},\dots,\lambda^{(N)}\big)$, $\lambda^{(i)}$ a~partition of~arbitrary length for~$i=1,\dots,N$, and
\begin{gather}
C_{N,\infty}(\vlam|\bs|q,t) =
\prod_{i=1}^N \prod_{i< j\leq k< \infty}
\frac{\big(q^{\lambda^{(i)}_{k-i+1} - \lambda^{(j)}_{k-j+1}}ts_{j}/s_i;q\big)_{\lambda^{(i)}_{k-i}-\lambda^{(i)}_{k-i+1}}}{\big(q^{\lambda^{(i)}_{k-i+1} - \lambda^{(j)}_{k-j+1}}qs_{j}/s_i;q\big)_{\lambda^{(i)}_{k-i}-\lambda^{(i)}_{k-i+1}}} \nonumber
\\ \hphantom{C_{N,\infty}(\vlam|\bs|q,t) = }
{} \times \prod_{i=1}^N \prod_{i\leq j<k<\infty}
\frac{\big(q^{-\lambda^{(j)}_{k-j}+\lambda^{(i)}_{k-i+1}}qs_{j}/ts_i ;q\big)_{\lambda^{(i)}_{k-i}-\lambda^{(i)}_{k-i+1}}}{\big(q^{-\lambda^{(j)}_{k-j}+\lambda^{(i)}_{k-i+1}}s_{j}/s_i ;q\big)_{\lambda^{(i)}_{k-i}-\lambda^{(i)}_{k-i+1}}}\label{CNinfty}
\end{gather}
setting $\lambda^{(i+N)}_j \equiv \lambda^{(i)}_j$.
\end{Lemma}
\begin{proof}
Straightforward computations, using that $\theta_{ik} = \lambda^{(i)}_{k-i}-\lambda^{(i)}_{k-i+1}$ defines a~one-to-one correspondence between multi-partitions $\vlam=\big(\lambda^{(1)},\dots,\lambda^{(N)}\big)$ in $\mathsf{P}^N$ and matrices $\theta=(\theta_{ik})_{i,k=1}^\infty$ in $\hat{\mathsf{M}}_N$ (the interested reader can find the details in Appendix~\ref{app:prooflem1}).
\end{proof}
It is interesting to note that
\begin{gather*}
\varphi_{N,\infty}(\bx|\bs|q,t) \equiv \prod_{i=1}^N\prod_{j>i}\frac{(qx_j/tx_i;q)_\infty}{(qx_j/x_i;q)_\infty}\, f_{N,\infty}(\bx|\bs|q,t)
\end{gather*}
is a~natural generalization of~the function in~\eqref{phiN} due to the following implication of~Theorem~\ref{thm1}.

\begin{Fact}
\label{fact:Id}
The following holds,
\begin{gather*}
\varphi^{\widehat{\mathfrak{gl}}_N}\big(p^{\delta/N}x,p^{1/N}|\kappa^{\delta/N} s,\kappa^{1/N}|q,q/t\big) = \varphi_{N,\infty}(\bx|\bs|q,t)
\end{gather*}
with the variables $x=(x_i)_{i=1}^N$ and $s=(x_i)_{i=1}^N$ on the left-hand side extended to variables $\bx=(x_i)_{i=1}^\infty$ and $\bs=(s_i)_{i=1}^\infty$ on the right-hand side by the rules in~\eqref{period1}.
Moreover, the conjectures in~\eqref{symmetry1} are equivalent to
\begin{alignat}{3}
& \varphi_{N,\infty}(\bx|\bs|q,t)= \varphi_{N,\infty}(\bs|\bx|q,t) \quad && \text{(bispectral duality)},&\nonumber\\
& \varphi_{N,\infty}(\bx|\bs|q,t)= \varphi_{N,\infty}(\bx|\bs|q,q/t) \quad && \text{(Poincar\'e duality)},& \label{symmetry2}
\end{alignat}
under the conditions in~\eqref{period1}.
\end{Fact}

\begin{proof}Since $p^{(j-i)/N}\big(p^{\delta/N}x\big)_j/\big(p^{\delta/N}x\big)_i=x_j/x_i$ for all $i,j=1,\dots,N$, we only need to show that
\begin{gather*}
\prod_{i=1}^N \prod_{j>i} \frac{(qx_j/tx_i;q)_\infty}{(qx_j/x_i;q)_\infty} =
\prod_{1\leq i<j\leq N}\frac{(qx_j/tx_i;q,p)_\infty}{(qx_j/x_i;q,p)_\infty}
\prod_{1\leq i\leq j\leq N}\frac{(qpx_i/tx_j;q,p)_\infty}{(qpx_i/x_j;q,p)_\infty}.
\end{gather*}
This is proved in Appendix~\ref{app:Id}, Lemma~\ref{lem:Id2}.
\end{proof}.

\subsection{Convergence}\label{sec:convergence}

We~prove that the non-stationary Ruijsenaars functions $f_N\!(x,p|s,\kappa|q,t)$ in Definitions~\ref{def2} and~\ref{def:fNinfty} are absolutely convergent in a~certain domain of~variables and parameters.

\begin{Theorem}\label{thm:converegence}
For fixed $N\in\Z_{\geq 1}$, assume that the variables $s=(s_1,\dots,s_N)\in\C^N$ and the parameters $q$ and $\kappa$ satisfy the following conditions,
\begin{itemize}\itemsep=0pt
\item[$(i)$] for some $\sigma>0$,
\begin{gather*}
{}|\sin\arg(s_i/s_j)|>\sigma \qquad (1\leq i<j\leq N),
\end{gather*}
\item[$(ii)$] $q$ and $\kappa$ both are real, and either $|q|<1$ and $|\kappa|>1$, or $|q|>1$ and $|\kappa|<1$.
\end{itemize}
Then, there exists a~constant $\rho>0$ such that the formal power series
\begin{gather*}
f_{N,\infty}(x,p;s,\kappa|q,t)\in \C[[x_2/x_1,\dots,x_N/x_{N-1},px_1/x_N]]
\end{gather*}
in Definitions~{\rm \ref{def2}} and~{\rm \ref{def:fNinfty}} is absolutely convergent in the domain
\begin{gather}\label{xcond}
|p|<\rho^N,\qquad |x_2/x_1|<\rho,\qquad \dots,\qquad |x_N/x_{N-1}|<\rho,\qquad |px_1/x_N|<\rho.
\end{gather}
\end{Theorem}
\begin{Remark}In our proof, we actually show convergence for any $\rho<1/C_1C_2$ where
\begin{gather}
C_1= 1+|1-t/q|\max\left(\frac1{\sigma}, \frac{|\kappa|}{|1-|\kappa||}\right)\!,\nonumber\\
C_2= 1+|1-q/t|\max\left(\frac1{\sigma},\frac{1}{|1-|q||}\right)\!.\label{C1C2}
\end{gather}
\end{Remark}
\begin{Remark}
We~believe that it~is possible to refine this convergence result. In particular, we~believe that there are regions of~convergence where $s_i/s_j$, $1\leq i<j\leq N$, are real and $q$ and~$\kappa$ have non-trivial imaginary parts.
\end{Remark}

\begin{proof}[Proof of~Theorem~\ref{thm:converegence}]
Our strategy of~proof is to show that our assumptions imply simple upper bounds on the terms appearing in the series in~\eqref{fNinfty1}--\eqref{CNinfty}:
\begin{gather}
\label{estim}
\Bigg|\prod_{i=1}^N\prod_{k\geq 1}(x_{i+k}/x_{i+k-1})^{\lambda^{(i)}_k}\Bigg|\leq \rho^{|\vlam|},\qquad \left|C_{N,\infty}(\vlam|\bs|q,t)\right|\leq C_1^{|\vlam|}C_2^{|\vlam|}
\end{gather}
with $|\vlam|\equiv \sum_{i=1}^N \sum_{k\geq 1}\lambda^{(i)}_k$ and $\alpha = \rho C_1C_2<1$.
With that, absolute convergence follows from the comparison test: the series in~\eqref{fNinfty1}--\eqref{CNinfty} is of~the form $\sum_{\vlam\in\mathsf{P}^N} a_{\vlam}$ with $|a_{\vlam}|\leq \alpha^{|\vlam|}$ for all $\vlam\in\mathsf{P}^N$,
and the series $\sum_{\vlam\in\mathsf{P}^N} \alpha^{|\vlam|}$ converges absolutely for~$|\alpha|<1$.

The first estimate in~\eqref{estim} is a~simple consequence of~the conditions in~\eqref{xcond}: since \allowbreak \mbox{$x_{i+N}=px_i$} for all $i\geq 1$, these conditions are equivalent to
\begin{gather*}
|x_{i+1}/x_i| < \rho\qquad (i\geq 1) ,
\end{gather*}
which clearly implies the result.

The proof of~the second estimate in~\eqref{estim} is more involved and, for this reason, we supplement our somewhat descriptive arguments in the main text below by a~detailed argument in~Appendix~\ref{app:estim}.

We~observe that $C_{N,\infty}(\vlam|\bs|q,t)$ in~\eqref{CNinfty} is a~product of~fractions $\big(1-q^{l}au\big)/\big(1-q^l u\big)$ with $a=t/q$ in the first group of~products and $a=q/t$ in the second group,
 $l\in\Z$, and $u=s_j/s_i$ for~$i=1,\dots,N$ and $j\geq i$; moreover, $s_{j+\ell N}=\kappa^\ell s_j$ for~$\ell\in\Z_{\geq 1}$.
Such a~fraction can be estimated in a~simple way:
\begin{gather*}
\left|\frac{1-q^{l}au}{1-q^l u}\right| = \left|1+(1-a)\frac{q^l u}{1-q^l u} \right|\leq 1+|1-a|\left|\frac{q^l u}{1-q^l u}\right|\!.
\end{gather*}
If $j-i$ is {\em not} an integer multiple of~$N$, we can estimate this further using
\begin{gather}
\label{zest}
\left| \frac{z}{1-z}\right| \leq \frac1{|\sin\arg(z)|} \qquad (z\in\C\setminus\{\R\})
\end{gather}
(to see that the latter inequality holds, write $z=|z|\ee^{\ii\varphi}$ and note that~\eqref{zest} is equivalent to
\begin{gather*}
|z|^2\sin^2\varphi \leq 1+|z|^2-2|z|\cos\varphi\Leftrightarrow 0\leq (1-|z|\cos\varphi)^2,
\end{gather*}
which is obvious). Since we assume that $q$ and $\kappa$ both are real,
\begin{gather*}
\big|\sin\arg\big(q^l \kappa^{\ell }s_j/s_i\big)\big| = |\sin\arg(s_j/s_i)| \geq \sigma>0\qquad (j-i\neq N\Z_{\geq 0})
\end{gather*}
for all integers $l,\ell$, we get a~simple universal bound for these fractions:
\begin{gather*}
\left|\frac{1-q^l a s_j/s_i}{1-q^l s_j/s_i}\right| \leq 1+|1-a|\frac{1}{\sigma} \qquad (j-i \notin N\Z_{\geq 0})
\end{gather*}
for all integers $l$. However, this bound does not work for~$j=i+\ell N$ with $\ell\in\Z_{\geq 0}$ since, in these cases, $q^l u=q^l s_j/s_i=q^l \kappa^\ell$ is real.
However, one can check that, in all these latter cases, either $l\leq 0$ and $\ell>0$, or $l<0$ and $\ell\geq 0$, and thus, by our assumptions, $z\equiv q^lu=q^l \kappa^\ell$ {\em always}
satisfies either $|z|\geq \min\big(|q|^{-1},|\kappa|\big)>1$ (if~$|q|<1$ and $|\kappa|>1$) or $|z|\leq \max\big(|q|,|\kappa|^{-1}\big)<1$ (if~$|q|>1$ and $|\kappa|<1$); we therefore can use the inequality
\begin{gather*}
\left| \frac{z}{1-z}\right| \leq \frac{|z|}{|1-|z||} \qquad (|z|\neq 1)
\end{gather*}
to get simple universal bounds for the cases $j=i+\ell N$ with $\ell\in\Z_{\geq 0}$ as well (we spell our the details of~this argument in Appendix~\ref{app:estb}).
We~thus get estimates
\begin{gather*}
\left|\frac{1-q^{l}au}{1-q^l u}\right| \leq C_{1,2}
\end{gather*}
with different upper bounds, $C_1$ and $C_2$, for {\em all} fractions in the first and second groups of~pro\-ducts on the right-hand side in~\eqref{CNinfty}, respectively.
The arguments above allow to compute the constants $C_1$ and $C_2$ and give the results in~\eqref{C1C2}; the interested reader can find the details of~this computation in Appendix~\ref{app:estim}.

Inserting these bounds into~\eqref{CNinfty} we obtain
\begin{gather}
|C_{N,\infty}(\vlam|\bs|q,t)| \leq
\prod_{i=1}^N \Bigg( \prod_{i< j\leq k< \infty} C_1^{\lambda^{(i)}_{k-i}-\lambda^{(i)}_{k-i+1}}\Bigg)\Bigg( \prod_{i\leq j< k< \infty} C_2^{\lambda^{(i)}_{k-i}-\lambda^{(i)}_{k-i+1}}\Bigg) \nonumber
\\ \hphantom{|C_{N,\infty}(\vlam|\bs|q,t)|}
{}= \prod_{i=1}^N \Bigg( \prod_{i< j< \infty} C_1^{\lambda^{(i)}_{j-i}}\Bigg)\Bigg( \prod_{i\leq j< \infty} C_2^{\lambda^{(i)}_{j+1-i}}\Bigg)\nonumber
\\ \hphantom{|C_{N,\infty}(\vlam|\bs|q,t)|}
{}= \prod_{i=1}^N \Bigg( \prod_{k\geq 1} C_1^{\lambda^{(i)}_{k}}\Bigg) \Bigg( \prod_{k\geq 1} C_2^{\lambda^{(i)}_{k}}\Bigg) = C_1^{|\vlam|}C_2^{|\vlam|},
\label{Cbound}
\end{gather}
computing telescoping products in the second step and using $\sum_{i=1}^N\sum_{k\geq 1}\lambda^{(i)}_k=|\vlam|$ in the last step.
This proves the second estimate in~\eqref{estim}.

To conclude, we prove that the series $\sum_{\vlam\in\mathsf{P}^N} \alpha^{|\vlam|}$ for~$|\alpha|<1$ is absolutely convergent by the following computation,
\begin{gather*}
\sum_{\vlam\in\mathsf{P}^N} \alpha^{|\vlam|} = \sum_{\lambda^{(1)},\dots,\lambda^{(N)}\in\mathsf{P}} \prod_{i=1}^N \alpha^{|\lambda^{(i)}|} =
 \prod_{i=1}^N \sum_{\lambda^{(i)}\in\mathsf{P}} \alpha^{|\lambda^{(i)}|} =\left( \sum_{\lambda\in\mathsf{P}} \alpha^{|\lambda|}\right)^N = \frac1{(\alpha;\alpha)_\infty^N},
\end{gather*}
using the definition $|\lambda|\equiv \sum_{k\geq 1}\lambda_k$ for partitions $\lambda$; for clarity, and for the convenience of~the reader, we give in Appendix~\ref{app:Id} the well-known identity used in the last step, together with its elementary proof making absolute convergence manifest; see~\eqref{Euler}--\eqref{Euler1}.
\end{proof}

\section[T-operators]{$\cT$-operators}\label{sec:cT}
For fixed $N\in\Z_{\geq 1}$, we define an operator $\cT$ which acts diagonally on the asymptotically free solution of~the trigonometric Ruijsenaars model (Section~\ref{sec:cTtrig}). We~also present a~natural non-stationary generalization of~this operator which, as we conjecture, acts diagonally on the corresponding non-stationary Ruijsenaars function (Section~\ref{sec:cTell}).

\subsection{Trigonometric case}\label{sec:cTtrig}
We~find it~convenient to work with formal power series.

\begin{Definition} For
\begin{gather}
\label{Delta}
\Delta \equiv \sum_{i=1}^N \left(x_i\partial_{x_i}+(N-i)\beta \right)^2
\end{gather}
with $\beta = \log(t)/\log(q)$, let
\begin{gather}
\label{cTtrig}
\cT_N(x|q,t) \equiv \sum_{\theta\in\mathsf{M}_N} \prod_{1\leq i<j\leq N}(x_j/x_i)^{\theta_{ij}} q^{\frac12\Delta}c_N(\theta|x|q,t)\prod_{1\leq i<j\leq N}\frac{(x_j/x_i;q)_\infty}{(tx_j/x_i;q)_\infty}
\end{gather}
on $x^\lambda \C[[x_2/x_1,\dots,x_N/x_{N-1}]]$ for~$\lambda\in\C^N$, with $\mathsf{M}_N$ in~\eqref{MMN} and $c_N(\theta|s|q,t)$ in~\eqref{cN}.
\end{Definition}

Clearly, the operator $\cT_N(x|q,t)$ is complicated: it~has the same complexity as the function $f_N(x|s|q,t)$; cf.~\eqref{fN}.
Still, it~is interesting since, different from the elliptic Macdonald--Ruijsenaars operators in~\eqref{R}, we know its natural generalization to the non-stationary case; see Section~\ref{sec:cTell}.

The following is our main result in this section.

\begin{Proposition}\label{prop:cT}
The $\cT$-operator in~\eqref{cTtrig} is well-defined, it~commutes with the trigonometric Macdonald--Ruijsenaars operators in~\eqref{MR}:
\begin{gather}\label{commutativity}
\big[\cT_N(x|q,t),D_N^\pm(x|q,t)\big] =0
\end{gather}
on $x^\lambda \C[[x_2/x_1,\dots,x_N/x_{N-1}]]$ for all $\lambda\in\C^N$, and it~acts diagonally on the asymptotically free solutions of~the trigonometric Ruijsenaars model in~\eqref{fN}--\eqref{cN}:
\begin{gather}
\cT_N(x|q,t) x^\lambda f_N(x|s|q,t)=\varepsilon_N(s|q) x^\lambda f_N(x|s|q,t),\qquad s_i=t^{N-i}q^{\lambda_i},
\\
\label{vepsN11}
\varepsilon_N(s|q) = q^{\frac12\sum_{i=1}^N[\log(s_i)/\log(q)]^2}
\end{gather}
$($note that $\log(s_i)/\log(q)=\lambda_i+\beta(N-i))$.
\end{Proposition}
(A proof based on results in the rest of~this section can be found in Appendix~\ref{app:cT}.)

Our proof of~Proposition~\ref{prop:cT} is based on the following convenient representation of~the \mbox{$\cT$-operator}.
\begin{Lemma}
\label{lem:cTtrig}
For $f(x)\in \C[[x_2/x_1,\dots,x_N/x_{N-1}]]$ and $\lambda\in\C^N$,
\begin{gather}
\label{cTtrig1}
\cT_N(x|q,t) x^\lambda f(x) = \varepsilon(\lambda)x^\lambda\Bigg[\prod_{i=1}^N\vartheta_3(s_ix_i/y_i|q) \chi_N(x|y|q,t)\!\!\prod_{1\leq i<j\leq N}\!\!(1-y_j/y_i)f(y)\Bigg]_{1,y}
\end{gather}
with
\begin{gather*}
\varepsilon(\lambda)= q^{\frac12\sum_{i^1}^N(\lambda_i+(N-i)\beta)^2},
\end{gather*}
$\vartheta_3(z|q)\equiv \sum_{n\in\Z}q^{\frac12 n^2}z^n$ the third Jacobi theta function,
\begin{gather}
\label{psiN}
\chi_N(x|y|q,t)\equiv f_N(x|y|q,t)\prod_{1\leq i<j\leq N}\frac{(qy_j/y_i;q)_\infty}{(ty_j/y_i;q)_\infty} ,
\end{gather}
and $[\cdots]_{1,y}$ is the constant term in $y$, i.e., for formal Laurent series $g(y)=\sum_{\mu\in\Z^N}g_\mu y^\mu$ as in~\eqref{cTtrig1}, $[g(y)]_{1,y}=g_0$.
\end{Lemma}

\begin{Remark}
We~use $[\cdots]_{1,y}$ only for~$g(y)\in \prod_{i=1}^N\vartheta_3(s_ix_i/y_i|q) \C[[y_2/y_1,\dots,y_N/y_{N-1}]]$,
and our definition of~$[\cdots]_{1,y}$ is non-ambiguous for these.
\end{Remark}

\begin{proof}[Proof of~Lemma~\ref{lem:cTtrig}] We~use that $\C[[x_2/x_1,\dots,x_N/x_{N-1}]]$ is spanned by (a subset of) monomials $x^\mu$ with $\mu\in\Z^N$.
For fixed $\lambda\in\C^N$, we compute the action of~$q^{\frac12\Delta}$ on $x^\lambda x^{\mu}$, $\mu\in\Z^N$:
\begin{gather*}
q^{\frac12\Delta} x^\lambda x^{\mu} = q^{\frac12\sum_{i=1}^N(\lambda_i+\mu_i+(N-i)\beta)^2}x^{\lambda +\mu} =\varepsilon(\lambda) x^\lambda \prod_{i=1}^N (x_is_i)^{\mu_i} q^{\frac12\sum_{i=1}^N\mu_i^2}
\end{gather*}
equal to
\begin{gather*}
 \varepsilon(\lambda) x^\lambda\left[\,\prod_{i=1}^N \vartheta_3(s_ix_i/y_i|q) y^\mu \right]_{1,y}\! ,
\end{gather*}
and thus
\begin{gather*}
q^{\frac12\Delta} x^\lambda f(x) = \varepsilon(\lambda) x^\lambda\left[\,\prod_{i=1}^N \vartheta_3(s_ix_i/y_i|q) f(y) \right]_{1,y}
\end{gather*}
for all $f(x)\in\C[[x_2/x_1,\dots,x_N/x_{N-1}]]$. This and the definition in~\eqref{cTtrig} give
\begin{gather*}
\cT_N(x|q,t) x^\lambda f(x) = \varepsilon(\lambda) x^\lambda \Biggl[\, \prod_{i=1}^N \vartheta_3(s_ix_i/y_i|q)
\\ \hphantom{\cT_N(x|q,t) x^\lambda f(x) =}
{}\times\sum_{\theta\in\mathsf{M}_N} \prod_{1\leq i<j\leq N}(x_j/x_i)^{\theta_{ij}} c_N(\theta|y|q,t)\prod_{1\leq i<j\leq N}\frac{(y_j/y_i;q)_\infty}{(ty_j/y_i;q)_\infty}f(y)\Biggr]_{1,y} ,
\end{gather*}
and using~\eqref{fN} and the definition in~\eqref{psiN} we obtained~\eqref{cTtrig1}.
\end{proof}

We~note that $\chi_N(x|y|q,t)= \chi_N(y|x|q,q/t)$ (this is proved in Appendix~\ref{app:cT}, Lemma~\ref{lem:psiN}); inserting this in~\eqref{cTtrig1} and backtracking, one obtains the following alternative representation of~the $\cT$-operator:
\begin{gather}
\label{cTtrig2}
\cT_N(x|q,t) \equiv \frac{(qx_j/x_i;q)_\infty}{(qx_j/tx_i;q)_\infty} \sum_{\theta\in\mathsf{M}_N}c_N(\theta|x|q,q/t)q^{\frac12\Delta}\!\!\!\prod_{1\leq i<j\leq N}\!\!\!(x_j/x_i)^{\theta_{ij}}\!\!\! \prod_{1\leq i<j\leq N}\!\!\!(1-x_j/x_i).
\end{gather}

\subsection{Non-stationary case}
\label{sec:cTell}
We~present a~non-stationary generalization of~the $\cT$-operator.

\begin{Definition}
\label{def:cTell}
For $\Delta$ in~\eqref{Delta} with $\beta = \log(t)/\log(q)$, let
\begin{gather}
\label{cTell}
\cT_{N,\infty}(x,p|q,t,\kappa) = \sum_{\theta\in\hat{\mathsf{M}}_N} \prod_{i=1}^N\prod_{j>i}(x_j/x_i)^{\theta_{ij}} q^{\frac12\Delta} \, T_{\kappa,p}c_{N,\infty}(\theta|\bx|q,t) \prod_{i=1}^N \prod_{j>i} \frac{(x_j/x_i;q)_\infty}{(tx_j/x_i;q)_\infty}
\end{gather}
with $x_{i+N}=p x_i$ for~$i\geq 1$ on $x^\lambda \C[[x_2/x_1,x_3/x_2,\dots,x_{N}/x_{N-1},px_1/x_N]]$ for~$\lambda\in\C^N$, with $\hat{\mathsf{M}}_N$ in~\eqref{hMMN} and $c_{N,\infty}(\theta|\bx|q,t)$ in~\eqref{cNinfty}.
\end{Definition}

\begin{Remark}
\label{rem:Id}
To make the $p$-dependence of~this operator manifest, one can write it~as
\begin{gather*}
\cT_{N,\infty}(x,p|q,t,\kappa) = \sum_{\theta\in\hat{\mathsf{M}}_N} e_{N,\infty}(x,p) q^{\frac12\Delta}\, T_{\kappa,p}c_{N,\infty}(\theta|x,p|q,t)
\\ \hphantom{\cT_{N,\infty}(x,p|q,t,\kappa) =}
{}\times
\prod_{1\leq i<j\leq N}\frac{(x_j/x_i;q,p)_\infty}{(tx_j/x_i;q,p)_\infty}\prod_{1\leq i\leq j\leq N}\frac{(px_i/x_j;q,p)_\infty}{(ptx_i/x_j;q,p)_\infty}
\end{gather*}
using the definitions in~\eqref{cNeNdef} and Lemma~\ref{lem:Id2} in Appendix~\ref{app:Id}.
\end{Remark}

By comparing with~\eqref{fNinfty}--\eqref{cNinfty}, it~is clear that the operator in~\eqref{cTell} is a~natural non-stationary generalization of~the trigonometric $\cT$-operators in~\eqref{cTtrig}; however, there is one important new feature: the shift operator $T_{\kappa,p}$ acting on $p$.

We~propose the following generalization to Proposition~\ref{prop:cT}; this conjecture is a~complement to the ones in \cite{S}.

\begin{Conjecture}
\label{conj:cTell}
The non-stationary $\cT$-operator in~\eqref{cTell} has a~well-defined diagonal action on the non-stationary Ruijsenaars function in Definitions~{\rm \ref{def1}} and~{\rm \ref{def:fNinfty}}:
\begin{gather*}
\cT_{N,\infty}(x,p|q,t,\kappa)x^\lambda f_{N,\infty}(x,p|s,\kappa|q,t) = \varepsilon_N(s|q)x^\lambda f_{N,\infty}(x,p|s,\kappa|q,t),\qquad
s_i = t^{N-i} q^{\lambda_i}
\end{gather*}
with $\varepsilon_N(s|q)$ given in~\eqref{vepsN11}.
\end{Conjecture}

In the rest of~this section, we present two generalizations of~results about the trigonometric $\cT$-operators: (i) the constant-term representation of~the $\cT$-operator in Lemma~\ref{lem:cTtrig}(ii) the alternative representation in~\eqref{cTtrig2} obtained with the duality in~\eqref{symmetry1}. We~also rephrase Conjecture~\ref{conj:cTell} in terms of~the non-stationary Ruijsenaars functions as defined in~\cite{S}.

One can adapt the proof Lemma~\ref{lem:cTtrig} to obtain the following constant-term representation of~the $\cT$-operator in~\eqref{cTell}:
\begin{Lemma}
For $f(x,p)\in \C[[x_2/x_1,px_3/x_2,\dots,x_{N}/x_{N-1},px_1/x_N]]$ and $\lambda\in\C^N$,
\begin{gather*}
\cT_{N,\infty}(x,p|q,t,\kappa)x^\lambda f(x,p)
\\ \qquad
{}= \varepsilon(\lambda)x^\lambda\Biggl[\,\prod_{i=1}^N\vartheta_3(s_ix_i/y_i|q)
\frac1{1-\kappa p/u}\chi_{N,\infty}(x,p|y,u|q,t)
 \prod_{i=1}^N\prod_{j >i} (1-y_j/y_i)f(y,u) \Biggr]_{1,y,u}
\end{gather*}
with
\begin{gather*}
\chi_{N,\infty}(x,p|y,u|q,t) = f_{N,\infty}(x,p|y,u|q,t) \prod_{i=1}^N \prod_{j>i} \frac{(qy_j/y_i;q)_\infty}{(ty_j/y_i;q)_\infty}
\end{gather*}
setting $y_{i+N}=u y_i$ for all $i\geq 1$, $f_{N,\infty}(x,p|y,u|q,t)$ in~\eqref{fNinfty}--\eqref{cNinfty}, and $[\cdots]_{1;y,u}$, short
for $[[\cdots]_{1,y}]_{1,u}$, the constant term in $y\in\C^N$ and $u\in\C$.
\end{Lemma}

\begin{proof}
This is proved by a~straightforward generalization of~the arguments given in the proof of~Lemma~\ref{lem:cTtrig}; the only new ingredient is
\begin{gather*}
T_{\kappa,p} p^n = (\kappa p)^n = \left[\frac1{1-\kappa p/u}u^n\right]_{1,u}\qquad (n\in\Z_{\geq 0}),
\end{gather*}
and therefore
\begin{gather*}
q^{\frac12\Delta}T_{\kappa,p} x^\lambda f(x,p) = \varepsilon(\lambda)x^\lambda \left[ \prod_{i=1}^N\vartheta_3(s_ix_i/y_i|q) \frac1{1-\kappa p/u} f(y,u) \right]_{1,y,u}
\end{gather*}
for all $f(x,p)\in \C[[x_2/x_1,px_3/x_2,\dots,x_{N}/x_{N-1},px_1/x_N]]$.
\end{proof}

Moreover, by an argument similar to the one for the trigonometric $\cT$-operator above, the conjectured duality in~\eqref{symmetry2} implies
\begin{gather*}
\chi_{N,\infty}(x,p|y,u|q,t)=\chi_{N,\infty}(y,u|x,p|q,q/t)
\end{gather*}
and the following alternative representation of~this $\cT$-operator
\begin{gather*}
\cT_{N,\infty}(x,p|q,t,\kappa)
= \prod_{i=1}^N \prod_{j>i} \frac{(qx_j/x_i;q)_\infty}{(qx_j/tx_i;q)_\infty} \sum_{\theta\in\hat{\mathsf{M}}_N} c_{N,\infty}(\theta|\bx|q,t)
q^{\frac12\Delta} T_{\kappa,p}
\\ \hphantom{\cT_{N,\infty}(x,p|q,t,\kappa)= }
{}\times \prod_{i=1}^N\prod_{j>i}(x_j/x_i)^{\theta_{ij}}
 \prod_{i=1}^N \prod_{j>i}(1-x_j/x_i)
\end{gather*}
with $x_{i+N}=px_i$ for all $i\geq 1$.

To conclude, we rephrase Conjecture~\ref{conj:cTell} using balanced coordinates.

\begin{Definition}\label{def:cTellB}
For $\Delta$ as in~\eqref{Delta} with $\beta=\log(q/t)/\log(q)$, let
\begin{gather}
\cT^{\widehat{\mathfrak{gl}}_N}(x,p|q,t,\kappa)\nonumber
\\[-1ex] \qquad
{}\equiv \sum_{\lambda^{(1)},\dots,\lambda^{(N)}\in\mathsf{P}}
\prod_{\beta=1}^N\prod_{\alpha\geq 1}(p x_{\alpha+\beta}/tx_{\alpha+\beta-1})^{\lambda^{(\beta)}_\alpha} q^{\frac12\Delta}T_{\kappa,p}
\prod_{i,j=1}^N\frac{\mathsf{N}^{(j-i|N)}_{\lambda^{(i)},\lambda^{(j)}}
(tx_j/x_i|q,p)}{\mathsf{N}^{(j-i|N)}_{\lambda^{(i)},\lambda^{(j)}}(x_j/x_i|q,p)}\nonumber
\\ \qquad\phantom{\equiv}
{}\times\prod_{1\leq i<j\leq N} \frac{\big(p^{j-i}x_j/x_i;q,p^N\big)_\infty}{\big(p^{j-i}qx_j/tx_i;q,p^N\big)_\infty} \prod_{1\leq i\leq j\leq N} \frac{\big(p^{N-j+i}x_i/x_j;q,p^N\big)_\infty}{\big(p^{N-j+i}qx_i/tx_j;q,p^N\big)_\infty}\label{cTell2}
\end{gather}
with $x_{i+N}=x_i$ for~$i\geq 1$ on $x^\lambda \C[[px_2/x_1,px_3/x_2,\dots,px_{N}/x_{N-1},px_1/x_N]]$ for~$\lambda\in\C^N$, with $\mathsf{P}$ the set of~all partitions and $\mathsf{N}^{(k|N)}_{\lambda,\mu}(u|q,\kappa)$
the Nekrasov factors given in~\eqref{N1}.
\end{Definition}

\begin{Fact}\label{conj:cTellB}Conjecture~\ref{conj:cTell} is equivalent to the following diagonal action of~the non-stationary $\cT$-operator in~\eqref{cTell2} on the non-stationary Ruijsenaars function in~\eqref{f1}--\eqref{N1},
\begin{gather}
\label{cTBeq}
\cT^{\widehat{\mathfrak{gl}}_N}(x,p|q,t)x^\lambda f^{\widehat{\mathfrak{gl}}_N}(x,p|s,\kappa|q,t) \,{=}\,\varepsilon_N(s|q)x^\lambda f^{\widehat{\mathfrak{gl}}_N}(x,p|s,\kappa|q,t),\quad
s_i \,{=}\, (q/t)^{N-i} q^{\lambda_i}.
\end{gather}
\end{Fact}

\begin{proof}This is implied by Theorem~\ref{thm1}, using that the shift operator $T_{\kappa,p}$ commutes with the following operator, $\Phi$, switching from unbalanced to balanced coordinates:
\begin{gather*}
(\Phi f)(x,p|s,\kappa|q,t)\equiv f\big(p^{-\delta}x,p^N|\kappa^{-\delta}s,\kappa^{N}|q,t/q\big),
\end{gather*}
and noting that Theorem~\ref{thm1} implies $(\Phi f_{N,\infty})(x,p|s,\kappa|q,t)= f^{\widehat{\mathfrak{gl}}_N}(x,p|s,\kappa|q,t)$.
\end{proof}

\section{Final remarks}\label{sec:remarks}

The main conjecture in \cite[Conjecture~1.14]{S} can be tested systematically using a~perturbative solution of~the elliptic Ruijsenaars model that generalizes the perturbative solution of~the elliptic Calogero--Sutherland (eCS) model in~\cite{L3}.
We~plan to present this elsewhere.

As already mentioned, one important outstanding problem is to find $\kappa$-deformations of~the elliptic Ruijsenaars operators in~\eqref{R} that have the non-stationary Ruijsenaars functions as eigenfunctions. As conjectured in \cite{S}, the limit $q\to 1$ of~this hypothetical non-stationary Ruijsenaars model is a~known non-stationary eCS model depending on parameters $\beta$, $p$ and $\kappa$ related to the non-stationary Ruijsenaars parameters as follows, $t=q^\beta$ and $\kappa=q^{-\varkappa}$.\footnote{The elliptic deformation parameter $p$ is the same in both cases.} Recently, a~rigorous construction of~integral representations of~eigenfunctions of~the non-stationary eCS model for~$\varkappa=\beta$ was presented~\cite{AL2}. We~hope that, by combining the latter results with recent results on the non-stationary Ruijsenaars functions for the corresponding special value of~$\kappa$~\cite{FOS}, it~will be possible to prove the main conjecture in \cite{S} in the non-stationary eCS limit $q\to 1$ and for~$\varkappa=\beta$.
Another possible strategy to prove the conjecture in~\cite{S} for~$q\to 1$ and general $\varkappa$-values is to try to generalize the perturbative solution of~the non-stationary Lam\'e equation in~\cite{AL1} (note that the latter equation reduces to the non-stationary eCS model for~$N=2$ in~special cases).

The elliptic Ruijsenaars model is invariant under the exchange $p\leftrightarrow q$ \cite{R87}.\footnote{We~thank S. Ruijsenaars for pointing this out at the NORDITA workshop ``Elliptic integrable systems, special functions and quantum field theory" in June 2019.}
The non-sta\-tionary Ruijsenaars functions do not have this property manifest; we plan to report elsewhere on how this duality is recovered from the non-stationary Ruijsenaars function.

It was suggested more than 20 years ago that the elliptic Ruijsenaars model has a~double-elliptic generalization with remarkable duality properties \cite{BMMM,FGNR}, and recently an explicit formula for an operator defining such a~model was conjectured \cite{KS}.
It would be interesting to obtain a~better understanding of~the relation between the non-stationary Ruijsenaars functions and this double elliptic system recently proposed in \cite{AKMM}.

Since the non-stationary $\cT$-operator proposed in this paper contains a~factor $q^{\frac12\Delta}T_{\kappa,p}$, its eigenvalue equation can be regarded as a~$q$-deformed heat equation. We~mention the work of~Felder and Varchenko on the $q$-deformed KZB heat equation \cite{FV2,FV3} which seems related; it~would be interesting to understand this relation in detail.

\appendix

\section{Alternative series representation}
\label{app:proof}
We~prove Theorem~\ref{thm1}.
We~start with details complementing the concise proof of~Lemma~\ref{lem1} in the main text (Appen\-dix~\ref{app:prooflem1}). The main part of~the proof is in Appendix~\ref{app:proofprop1}.

\subsection[Details on Lemma~3.5]{Details on Lemma~\ref{lem1}}\label{app:prooflem1}
One can check that the following two formulas provide a~correspondence between multi-partitions $\vlam=\big(\lambda^{(1)},\dots,\lambda^{(N)}\big)$ in $\mathsf{P}^N$ and matrices $\theta=(\theta_{ik})_{i,k=1}^\infty$ in $\hat{\mathsf{M}}_N$ that is one-to-one:
\begin{gather*}
\theta_{ik} = \lambda^{(i)}_{k-i}-\lambda^{(i)}_{k-i+1}
\end{gather*}
and
\begin{gather}\label{lamfromtet}
\lambda^{(i)}_{k-i} = \sum_{a\geq k}\theta_{ia}
\end{gather}
setting $\lambda^{(i+N)}=\lambda^{(i)}$ for all $i\geq 1$. With this identification, $C_{N,\infty}(\vlam|\bs|q,t)$ in~\eqref{CNinfty} is clearly equal to $c_{N,\infty}(\theta|\bs|q,t)$ in~\eqref{cNinfty}, and
\begin{gather*}
\begin{split}
&\prod_{i=1}^N \prod_{k=i+1}^\infty (x_k/x_i)^{\theta_{ik}} = \prod_{i=1}^N \prod_{k=i+1}^\infty [(x_k/x_{k-1})(x_{k-1}/x_{k-2})\cdots (x_{i+1}/x_i)]^{\theta_{ik}}
\\ &\hphantom{\prod_{i=1}^N \prod_{k=i+1}^\infty (x_k/x_i)^{\theta_{ik}}}
{}= \prod_{i=1}^N \prod_{i<j\leq k<\infty} (x_j/x_{j-1})^{ \lambda^{(i)}_{k-i}-\lambda^{(i)}_{k-i+1}}
 = \prod_{i=1}^N \prod_{j>i} (x_j/x_{j-1})^{\lambda^{(i)}_{j-i}}
\\ &\hphantom{\prod_{i=1}^N \prod_{k=i+1}^\infty (x_k/x_i)^{\theta_{ik}}}
{}= \prod_{i=1}^N \prod_{k\geq 1}(x_{i+k}/x_{i+k-1})^{\lambda^{(i)}_k},
\end{split}
\end{gather*}
inserting a~telescoping product in the second step, using~\eqref{lamfromtet} in the third, and computing a~telescoping product in the fourth. This proves the result.

\subsection[Proof of~Theorem~3.2]{Proof of~Theorem~\ref{thm1}}
\label{app:proofprop1}
We~show by direct computations that the function on the left-hand side in~\eqref{cor1} is equal to the function $f_{N,\infty}(\bx|\bs|q,t)$ in~\eqref{fNinfty1}--\eqref{CNinfty} with $x_{i+N}=px_i$ and $s_{i+N}=\kappa s_i$, for all $i\geq 1$.
This, together with Lemma~\ref{lem1}, proves the result.

We~compute the function on the left-hand side in~\eqref{cor1} using~\eqref{f1}--\eqref{N1}:
\begin{gather*}
f^{\widehat{\mathfrak{gl}}_N}\big(p^{\delta/N} x,p^{1/N}|\kappa^{\delta/N} s,\kappa^{1/N}|q,q/t\big)
\\ \qquad
{}=
\sum_{\lambda^{(1)},\dots,\lambda^{(N)}\in\mathsf{P}} \prod_{i,j=1}^N\frac{\mathsf{N}^{(j-i|N)}_{\lambda^{(i)},\lambda^{(j)}}\big((q/t)\big(\kappa^{\delta/N} s\big)_j/(\kappa^{\delta/N} s)_i|q,\kappa^{1/N}\big)}{\mathsf{N}^{(j-i|N)}_{\lambda^{(i)},\lambda^{(j)}}\big(\big(\kappa^{\delta/N} s\big)_j/\big(\kappa^{\delta/N} s\big)_i|q,\kappa\big)}
\\ \phantom{\qquad=}
{}\times
\prod_{\beta=1}^N\prod_{\alpha\geq 1}\big(p^{1/N}\big(p^{\delta/N}\big)_{\alpha+\beta}t/q
\big(p^{\delta/N}\big)_{\alpha+\beta-1}\big)^{\lambda^{(\beta)}_\alpha}.
\end{gather*}
By Definition~\ref{def1} of~the non-stationary Ruijsenaars functions, the variables $\big(p^{\delta/N} x\big)_i$ above are extended from $i=1,\dots,N$ to all $i\geq 1$ by the rule $\big(p^{\delta/N} x\big)_{i+N}=\big(p^{\delta/N} x\big)_i$, whereas $x_{i+N}=px_i$ for all $i\geq 1$ implies $\big(p^{\delta/N} x\big)_{i+kN}=p^{(N-i)/N}x_i = p^{[N-(i+k N)]/N}x_{i+k N}$ for all $i=1,\dots,N$ and~$k\in\Z_{\geq 1}$, and thus
\begin{gather*}
\big(p^{\delta/N} x\big)_i = p^{(N-i)/N}x_i\qquad (i\geq 1).
\end{gather*}
Therefore,
\begin{gather*}
\prod_{\beta=1}^N\prod_{\alpha\geq 1}\big(p^{1/N}\big(p^{\delta/N} x\big)_{\alpha+\beta}t/q\big(p^{\delta/N} x\big)_{\alpha+\beta-1}\big)^{\lambda^{(\beta)}_\alpha}
\\ \qquad
{}= \prod_{\beta=1}^N\prod_{\alpha\geq 1}\big(p^{1/N}p^{(N-\alpha-\beta)/N}x_{\alpha+\beta}t/q p^{(N-\alpha-\beta+1)/N}x_{\alpha+\beta-1}\big)^{\lambda^{(\beta)}_\alpha}
\\ \qquad
{}= (t/q)^{|\vlam|} \prod_{\beta=1}^N\prod_{\alpha\geq 1}(x_{\alpha+\beta}/x_{\alpha+\beta-1})^{\lambda^{(\beta)}_\alpha}
\end{gather*}
using the abbreviation $|\vlam|\equiv \sum_{\beta=1}^N \sum_{\alpha\geq 1}\lambda^{(\beta)}_\alpha$.
Renaming indices $(\alpha,\beta)\to (k,i)$, we thus can write the function on the left-hand side in~\eqref{cor1} as
\begin{gather*}
f^{\widehat{\mathfrak{gl}}_N}\big(p^{\delta/N} x,p^{1/N}|\kappa^{\delta/N} s,\kappa^{1/N}|q,q/t\big) = \sum_{\vlam\in\mathsf{P}^N} \tilde C_N(\vlam;s|q,t,\kappa)
 \prod_{i=1}^N\prod_{k\geq 1}(x_{i+k}/x_{i+k-1})^{\lambda^{(i)}_k}
\end{gather*}
with $x_{i+N}=px_i$ for all $i\geq 1$ and
\begin{gather}
\label{tCN}
 \tilde C_N(\vlam;s|q,t,\kappa) = (t/q)^{|\vlam|}\prod_{i,j=1}^N\frac{\mathsf{N}^{(j-i|N)}_{\lambda^{(i)},\lambda^{(j)}}\big((q/t)\big(\kappa^{\delta/N} s\big)_j/\big(\kappa^{\delta/N} s\big)_i|q,\kappa^{1/N}\big)}{\mathsf{N}^{(j-i|N)}_{\lambda^{(i)},\lambda^{(j)}}\big(\big(\kappa^{\delta/N} s\big)_j/\big(\kappa^{\delta/N} s\big)_i|q,\kappa^{1/N}\big)} .
\end{gather}
To complete the proof, we have to show that $\tilde C_N(\vlam;s|q,t,\kappa)$ in~\eqref{tCN} is equal to $C_{N,\infty}(\vlam;\bs|q,t)$ in~\eqref{CNinfty} for~$s_{i+N}=\kappa s_i$ ($i\geq 1$).
For that, we compute the Nekrasov factors in~\eqref{N1}, partially specializing to the variables we need:
\begin{gather*}
\mathsf{N}^{(j-i|N)}_{\lambda^{(i)},\lambda^{(j)}}\big(u|q,\kappa^{1/N}\big)=
\prod_{\setatop{b\geq a\geq 1}{b-a\equiv (j-i)(\mathrm{mod}N)}}
\big(u q^{-\lambda^{(j)}_a+\lambda^{(i)}_{b+1}}\kappa^{(-a+b)/N};q\big)_{\lambda^{(i)}_{b}-\lambda^{(i)}_{b+1}}
\\ \hphantom{\mathsf{N}^{(j-i|N)}_{\lambda^{(i)},\lambda^{(j)}}\big(u|q,\kappa^{1/N}\big)= }
{} \times \prod_{\setatop{\beta\geq \alpha\geq 1}{\beta-\alpha\equiv (i-j-1)(\mathrm{mod}N)}}\big(uq^{\lambda^{(i)}_\alpha-\lambda^{(j)}_\beta}
\kappa^{(\alpha-\beta-1)/N};q\big)_{\lambda^{(j)}_\beta-\lambda^{(j)}_{\beta+1}} .
\end{gather*}
We~note that the constraints on $b$ in the first product is solved by $b=a+j+\ell N-i$ with $\ell$ an arbitrary integer $\geq \chi(i>j)$, using the definition $\chi(i>j)=1$ for~$j<i$ and $0$ otherwise; similarly, the constraints on $\beta$ in the second product is solved by $\beta=\alpha+i+\ell' N-j-1$ with arbitrary integer $\ell'\geq \chi(j\geq i)$. We~thus can write these Nekrasov factors as
\begin{gather*}
\mathsf{N}^{(j-i|N)}_{\lambda^{(i)},\lambda^{(j)}}\big(u|q,\kappa^{1/N}\big) = \prod_{a\geq 1}\prod_{\ell\geq \chi(i>j)}\!\!\!\!\big(u q^{-\lambda^{(j)}_a+\lambda^{(i)}_{a+j+\ell N-i+1}}\kappa^{(j+\ell N-i)/N};q\big)_{\lambda^{(i)}_{a+j+\ell N-i}-\lambda^{(i)}_{a+j+\ell N-i+1}}
\\ \hphantom{\mathsf{N}^{(j-i|N)}_{\lambda^{(i)},\lambda^{(j)}}\big(u|q,\kappa^{1/N}\big) }
{}\times
\prod_{\alpha\geq 1}\prod_{\ell'\geq \chi(j\geq i)} \!\!\!\! \!\! \big(uq^{\lambda^{(i)}_\alpha-\lambda^{(j)}_{\alpha+i+\ell' N-j-1}}\kappa^{(j-i-\ell' N)/N};q\big)_{\lambda^{(j)}_{\alpha+i+\ell' N-j-1}-\lambda^{(j)}_{\alpha+i+\ell' N-j}}\! .
\end{gather*}
We~now specialize further to the arguments of~interest to us:
\begin{gather*}
u = c\big(\kappa^{\delta/N} s\big)_j/\big(\kappa^{\delta/N} s\big)_i = c \kappa^{(N-j)/N}s_j/\kappa^{(N-i)/N}s_i = c \kappa^{(i-j)/N} s_j/s_i ,\qquad c\in\{1,q/t\}.
\end{gather*}
For these arguments, the manifest $\kappa$-dependence disappears:
\begin{gather*}
\mathsf{N}^{(j-i|N)}_{\lambda^{(i)},\lambda^{(j)}}\big(c\big(\kappa^{\delta/N} s\big)_j/\big(\kappa^{\delta/N} s\big)_i|q,\kappa^{1/N}\big)
\\ \qquad
= \prod_{a\geq 1}\prod_{\ell\geq \chi(i>j)}\big(c q^{-\lambda^{(j)}_a+\lambda^{(i)}_{a+j+\ell N-i+1}}s_{j+\ell N}/s_i;q\big)_{\lambda^{(i)}_{a+j+\ell N-i}-\lambda^{(i)}_{a+j+\ell N-i+1}}
\\ \qquad\phantom{=}
{}\times \prod_{\alpha\geq 1}\prod_{\ell'\geq \chi(j\geq i)}\big(cq^{\lambda^{(i)}_\alpha-\lambda^{(j)}_{\alpha+i+\ell' N-j-1}}s_j/s_{i+\ell'N};q\big)_{\lambda^{(j)}_{\alpha+i+\ell' N-j-1}-\lambda^{(j)}_{\alpha+i+\ell' N-j}}
\end{gather*}
using $\kappa^\ell s_j/s_i=s_{j+\ell N}/s_i$ and $\kappa^{-\ell'} s_j/s_i=s_j/s_{i+\ell' N}$ implied by $s_{i+N}=\kappa s_i$ for~$i\geq 1$.
We~now take the product of~these Nekrasov factors over $i,j=1,\dots,N$, change variables $j+\ell N\to j$ in~the first group of~products and $i+\ell' N\to i$ in the second group, and obtain
\begin{gather*}
\prod_{i,j=1}^N \mathsf{N}^{(j-i|N)}_{\lambda^{(i)},\lambda^{(j)}}\big(c\big(\kappa^{\delta/N} s\big)_j/\big(\kappa^{\delta/N} s\big)_i|q,\kappa^{1/N}\big)
\\ \qquad
{}= \prod_{i=1}^N \prod_{j\geq i} \prod_{a\geq 1}\big(c q^{-\lambda^{(j)}_a+\lambda^{(i)}_{a+j-i+1}}s_{j}/s_i;q\big)_{\lambda^{(i)}_{a+j-i}-\lambda^{(i)}_{a+j-i+1}} \\ \qquad\phantom{=}
{}\times\prod_{j=1}^N \prod_{i>j} \prod_{\alpha\geq 1}\big(cq^{\lambda^{(i)}_\alpha-\lambda^{(j)}_{\alpha+i-j-1}}s_j/s_{i};q\big)_{\lambda^{(j)}_{\alpha+i-j-1}-\lambda^{(j)}_{\alpha+i-j}} \\ \qquad
{}= \prod_{i=1}^N \prod_{j\geq i} \prod_{k>j}\big(c q^{-\lambda^{(j)}_{k-j}+\lambda^{(i)}_{k-i+1}}s_{j}/s_i;q\big)_{\lambda^{(i)}_{k-i}-\lambda^{(i)}_{k-i+1}}
\\ \qquad\phantom{=}
{}\times\prod_{j=1}^N \prod_{i>j} \prod_{k\geq i}\big(cq^{\lambda^{(i)}_{k-i+1}-\lambda^{(j)}_{k-j}}s_j/s_{i};q\big)_{\lambda^{(j)}_{k-j}
-\lambda^{(j)}_{k-j+1}},
\end{gather*}
where we changed variables $a\to k=a+j$ and $\alpha\to k=\alpha+i-1$ in the last step.
We~find it~convenient to write this result as
\begin{gather*}
\prod_{i,j=1}^N \! \mathsf{N}^{(j-i|N)}_{\lambda^{(i)},\lambda^{(j)}}\big(c\big(\kappa^{\delta/N} s\big)_j/\big(\kappa^{\delta/N} s\big)_i|q,\kappa^{1/N}\big)
= \prod_{i=1}^N \prod_{i\leq j<k<\infty} \!\!\!\!\!\big(c q^{-\lambda^{(j)}_{k-j}+\lambda^{(i)}_{k-i+1}}s_{j}/s_i;q\big)_{\lambda^{(i)}_{k-i}-\lambda^{(i)}_{k-i+1}} \\ \hphantom{\prod_{i,j=1}^N \mathsf{N}^{(j-i|N)}_{\lambda^{(i)},\lambda^{(j)}}(c(\kappa^{\delta/N} s)_j/(\kappa^{\delta/N} s)_i|q,\kappa^{1/N})= }
{}\times
\prod_{i=1}^N \prod_{i<j\leq k<\infty} \!\!\! \!\! \big(cq^{\lambda^{(j)}_{k-j+1}-\lambda^{(i)}_{k-i}}s_i/s_{j};q\big)_{\lambda^{(i)}_{k-i}-\lambda^{(i)}_{k-i+1}}
\end{gather*}
swapping variable names $i\leftrightarrow j$ in the second group of~products.
We~insert this into~\eqref{tCN} to~obtain
\begin{gather*}
 \tilde C_N(\vlam;s|q,t,\kappa) = (t/q)^{|\vlam|}\prod_{i=1}^N \prod_{i\leq j<k<\infty} \frac{\big(q^{-\lambda^{(j)}_{k-j}+\lambda^{(i)}_{k-i+1}}qs_{j}/ts_i;q\big)_{\lambda^{(i)}_{k-i}
 -\lambda^{(i)}_{k-i+1}}}{\big(q^{-\lambda^{(j)}_{k-j}+\lambda^{(i)}_{k-i+1}}s_{j}/s_i;q\big)_{\lambda^{(i)}_{k-i}
 -\lambda^{(i)}_{k-i+1}}}
 \\ \hphantom{ \tilde C_N(\vlam;s|q,t,\kappa) =}
{}\times
\prod_{i=1}^N \prod_{i<j\leq k<\infty} \frac{\big(q^{\lambda^{(j)}_{k-j+1}-\lambda^{(i)}_{k-i}}qs_i/ts_{j};q\big)_{\lambda^{(i)}_{k-i}
-\lambda^{(i)}_{k-i+1}}}{\big(q^{\lambda^{(j)}_{k-j+1}-\lambda^{(i)}_{k-i}}s_i/s_{j};q\big)_{\lambda^{(i)}_{k-i}
-\lambda^{(i)}_{k-i+1}}}.
\end{gather*}
To proceed, we use the well-known identity
\begin{gather*}
\frac{(q^{-m}/a;q)_m}{(q^{-m}/b;q)_m} = (b/a)^m\frac{(qa;q)_m}{(qb;q)_m} \qquad
(a,b\in\C,\ m\in\Z_{\geq 0}).
\end{gather*}
Applying this to the factors in the second group of~products for~$m=\lambda^{(i)}_{k-i}-\lambda^{(i)}_{k-i+1}$, $b=q^{\lambda^{(i)}_{k-i+1}-\lambda^{(j)}_{k-j+1}}s_j/s_i$, $a=tb/q$ yields
\begin{gather*}
\tilde C_N(\vlam;s|q,t,\kappa) = (t/q)^{|\vlam|}\prod_{i=1}^N \prod_{i\leq j<k<\infty} \frac{\big(q^{-\lambda^{(j)}_{k-j}+\lambda^{(i)}_{k-i+1}}qs_{j}/ts_i;q\big)_{\lambda^{(i)}_{k-i}
-\lambda^{(i)}_{k-i+1}}}{\big(q^{-\lambda^{(j)}_{k-j}+\lambda^{(i)}_{k-i+1}}s_{j}/s_i;q\big)_{\lambda^{(i)}_{k-i}-\lambda^{(i)}_{k-i+1}}} \\ \hphantom{\tilde C_N(\vlam;s|q,t,\kappa) =}
{}\times
\prod_{i=1}^N \prod_{i<j\leq k<\infty} (q/t)^{\lambda^{(i)}_{k-i}-\lambda^{(i)}_{k-i+1}} \frac{\big(q^{\lambda^{(i)}_{k-i+1}-\lambda^{(j)}_{k-j+1}}ts_j/s_i\big)_{\lambda^{(i)}_{k-i}
-\lambda^{(i)}_{k-i+1}}}{\big(q^{\lambda^{(i)}_{k-i+1}-\lambda^{(j)}_{k-j+1}}qs_j/s_i\big)_{\lambda^{(i)}_{k-i}-\lambda^{(i)}_{k-i+1}}}.
\end{gather*}
To complete the proof that $ \tilde C_N(\vlam;s|q,t,\kappa)$ in~\eqref{tCN} is identical with $C_{N,\infty}(\vlam|\bs|q,t)$ in~\eqref{CNinfty}, we~swap the order of~the two groups of~products and compute the overall power of~$(q/t)$:
\begin{gather*}
\prod_{i=1}^N \prod_{i<j\leq k<\infty} (q/t)^{\lambda^{(i)}_{k-i}-\lambda^{(i)}_{k-i+1}} = \prod_{i=1}^N\prod_{j=i+1}^\infty\prod_{k=j}^\infty (q/t)^{\lambda^{(i)}_{k-i}-\lambda^{(i)}_{k-i+1}}
\\ \hphantom{\prod_{i=1}^N \prod_{i<j\leq k<\infty} (q/t)^{\lambda^{(i)}_{k-i}-\lambda^{(i)}_{k-i+1}}}
{}= \prod_{i=1}^N \prod_{j=i+1}^\infty (q/t)^{\lambda^{(i)}_{j-i}} = \prod_{i=1}^N\prod_{k\geq 1} (q/t)^{\lambda^{(i)}_{k}} =(q/t)^{|\vlam|},
\end{gather*}
cancelling the factor $(t/q)^{|\vlam|}$.
This proves the identity in~\eqref{Identity1} with $f_{N,\infty}(\bx|\bs|q,t)$ in~\eqref{fNinfty1}--\eqref{CNinfty} and $x_{i+N}=px_i$, $s_{i+N}=\kappa s_i$ ($i\geq 1$).
This, together with Lemma~\ref{lem1}, implies the result.

\section{Estimates}\label{app:estim}
{\sloppy
We~give a~complementary proof of~the second estimate in~\eqref{estim}, to compute the upper bounds~$C_{1,2}$ in Theorem~\ref{thm:converegence}.

}

\subsection[Complementary proof of~the second estimate in (3.13)]{Complementary proof of~the second estimate in~(\ref{estim})}
\label{app:proofestim}
We~prove that, under the assumptions in Theorem~\ref{thm:converegence}, the following estimates hold true for the fractions appearing in the formula~\eqref{CNinfty} for~$C_{N,\infty}(\vlam|\bs|q,t)$,
\begin{gather}
\label{frac1}
\left|\frac{\big(q^{\lambda^{(i)}_{k-i+1} - \lambda^{(j)}_{k-j+1}}ts_{j}/s_i;q\big)_{\lambda^{(i)}_{k-i}-\lambda^{(i)}_{k-i+1}}}{\big(q^{\lambda^{(i)}_{k-i+1} - \lambda^{(j)}_{k-j+1}}qs_{j}/s_i;q\big)_{\lambda^{(i)}_{k-i}-\lambda^{(i)}_{k-i+1}}}\right|\leq C_1^{\lambda^{(i)}_{k-i}-\lambda^{(i)}_{k-i+1}} \ \ (1\leq i\leq N,\ i<j\leq k<\infty) ,\!\!\!
\\
\label{frac2}
\left| \frac{\big(q^{-\lambda^{(j)}_{k-j}+\lambda^{(i)}_{k-i+1}}qs_{j}/ts_i ;q\big)_{\lambda^{(i)}_{k-i}-\lambda^{(i)}_{k-i+1}}}{\big(q^{-\lambda^{(j)}_{k-j}+\lambda^{(i)}_{k-i+1}}s_{j}/s_i ;q\big)_{\lambda^{(i)}_{k-i}-\lambda^{(i)}_{k-i+1}}} \right| \leq C_2^{\lambda^{(i)}_{k-i}-\lambda^{(i)}_{k-i+1}} \ \
(1\leq i\leq N,\ i\leq j< k<\infty)\!\!
\end{gather}
for all $N$-partitions $\vlam=\big(\lambda^{(1)},\dots,\lambda^{(N)}\big)\in\mathsf{P}^N$, with $C_1$ and $C_2$ in~\eqref{C1C2}.
This and~\eqref{CNinfty} imply the estimate in~\eqref{Cbound} which, by the computation in~\eqref{Cbound}, is equivalent to the second estimate in~\eqref{estim}.

We~observe all estimates in~\eqref{frac1}--\eqref{frac2} are of~the form
\begin{gather*}
\left|\frac{\big(q^{l}as_i/s_j;q\big)_\theta}{\big(q^l s_i/s_j;q\big)_\theta}\right|\leq C^\theta,
\end{gather*}
where
\begin{gather}
\label{first}
l= \lambda^{(i)}_{k-i+1} - \lambda^{(j)}_{k-j+1}+1,\qquad
a=t/q,\qquad
\theta= \lambda^{(i)}_{k-i}-\lambda^{(i)}_{k-i+1},\qquad
C=C_1
\end{gather}
in~\eqref{frac1} and
\begin{gather}
\label{second}
l=-\lambda^{(j)}_{k-j}+\lambda^{(i)}_{k-i+1},\qquad
a=q/t,\qquad
\theta= \lambda^{(i)}_{k-i}-\lambda^{(i)}_{k-i+1},\qquad
C=C_2
\end{gather}
in~\eqref{frac2}. We~prove~\eqref{frac1}--\eqref{frac2} using three different kinds of~estimates:

\begin{Lemma}\label{lem:est}
Let $\theta\in\Z_{\geq 0}$, $a\in\C$, $q, \kappa\in\R$ with either $|q|<1$ and $|\kappa|>1$ or $|q|>1$ and $|\kappa|<1$. Then the following estimates hold true,
\begin{itemize}\itemsep=0pt
\item[$(a)$] for all $l\in\Z$ and $u\in\C\setminus\{\R\}$:
\begin{gather}\label{esta}
\left| \frac{\big(q^{l}au;q\big)_\theta}{\big(q^l u;q\big)_\theta}\right|\leq \left( 1 + \frac{|1-a|}{|\sin\arg(u)|}\right)^\theta\!,
\end{gather}
\item[$(b)$] for all $m\in\Z_{\geq 0}$, $\ell\in\Z_{\geq 1}$:
\begin{gather}\label{estb}
\left|\frac{\big( q^{-\theta-m+1} a \kappa^\ell \big)_\theta}{\big(q^{-\theta-m+1} \kappa^\ell \big)_\theta} \right| \leq \left(1+|1-a|\frac{|\kappa|}{|1-|\kappa||} \right)^\theta\!,
\end{gather}
\item[$(c)$] for all $m\in\Z_{\geq 0}$, $\ell\in\Z_{\geq 0}$:
\begin{gather}\label{estc}
\left|\frac{\big( q^{-\theta-m} a \kappa^\ell \big)_\theta}{\big(q^{-\theta-m} \kappa^\ell \big)_\theta} \right| \leq \left(1+|1-a|\frac{1}{|1-|q||} \right)^\theta\!.
\end{gather}
\end{itemize}
\end{Lemma}
(The proof is given in Appendix~\ref{app:prooflem:est}.)

{\bf Case A:} For $j-i\notin N\Z_{\geq 0}$, we can use the estimate in~\eqref{esta}:
Since $s_{j+N}=\kappa s_j$ and~$\kappa$ is real, we have $|\sin\arg(s_j/s_i)|=|\sin\arg(s_{j+N}/s_i)|=|\sin\arg(s_i/s_j)|$ for all $j\geq i$;
since $|\sin\arg(s_j/s_i)|\geq \sigma$ for all $1\leq i<j\leq N$ by assumption, $|\sin\arg(s_j/s_i)|\geq \sigma$ for all $1\leq i\leq N$ and $j\geq i$ such that $j-i\neq N\Z_{\geq 0}$, and we get
\begin{gather*}
\left|\frac{(q^{l}as_i/s_j;q)_\theta}{(q^l s_i/s_j;q)_\theta}\right| \leq \left(1+\frac{|1-a|}{\sigma} \right)^\theta\qquad
(j-i\notin N\Z_{\geq 0})
\end{gather*}
for all cases in~\eqref{first}--\eqref{second}. This proves that the estimates in~\eqref{frac1}--\eqref{frac2} for all
\begin{gather}
\label{C1C2A}
C_1\geq 1+\frac{|1-t/q|}{\sigma},\qquad
C_2\geq 1+\frac{|1-q/t|}{\sigma}
\end{gather}
and for all cases $j-i\notin N\Z_{\geq 0}$.

We~consider the remaining cases for~\eqref{frac1} and~\eqref{frac2} below in Cases B and C, respectively.

{\bf Case B:} For $j-i\in N\Z_{\geq 1}$, we have $s_j/s_i = s_{i+\ell N}/s_i=\kappa^\ell$ for some $\ell\in\Z_{\geq 1}$, and we can use the estimate in~\eqref{estb}:
\begin{gather*}
\left|\frac{\big(q^{-\theta-m+1}as_i/s_j;q\big)_\theta}{\big(q^{-\theta-m+1} s_i/s_j;q\big)_\theta}\right|\leq \left(1+|1-a|\frac{|\kappa|}{||\kappa|-1|} \right)^\theta\qquad
(j-i\in N\Z_{\geq 1},\ m\in\Z_{\geq 0}).
\end{gather*}
We~check that all cases in~\eqref{frac1} for~$j-i\in N\Z_{\geq 1}$ are covered by this: all $l$ in~\eqref{first} for~$j=i+\ell N$ can be written as (recall that $\lambda^{(i+\ell N)}_k=\lambda^{(i)}_k$)
\begin{gather*}
l = -\big[\lambda^{(i)}_{k-i}- \lambda^{(i)}_{k-i+1}\big] - \big[\lambda^{(i)}_{k-i-\ell N+1}-\lambda^{(i)}_{k-i}\big]+1 = -\theta-m+1
\end{gather*}
with $m=\lambda^{(i)}_{k-i-\ell N+1}-\lambda^{(i)}_{k-i}\geq 0$ since $\ell\geq 1$ and $\lambda^{(i)}=\big(\lambda_1^{(i)},\lambda_2^{(i)},\dots\big)$ is a~partition.
This proves that~\eqref{frac1} holds true if
\begin{gather}
\label{C1C2B}
C_1\geq 1+|1-t/q|\frac{|\kappa|}{|1-|\kappa||}
\end{gather}
for all cases $j-i\in N\Z_{\geq 1}$.

{\bf Case C:} For $j-i\in N\Z_{\geq 0}$, we have $s_j/s_i = s_{i+\ell N}/s_i=\kappa^\ell$ for some $\ell\in\Z_{\geq 0}$, and we can use the estimate in~\eqref{estc}:
\begin{gather*}
\left|\frac{\big(q^{-\theta-m}as_i/s_j;q\big)_\theta}{\big(q^{-\theta-m} s_i/s_j;q\big)_\theta}\right|\leq \left(1+|1-a|\frac{1}{|1-|q||} \right)^\theta\qquad
(j-i\in N\Z_{\geq 0},\ m\in\Z_{\geq 0}).
\end{gather*}
We~check that all cases in~\eqref{frac2} for~$j-i\in N\Z_{\geq 0}$ are covered by this: all $l$ in can be written as
\begin{gather*}
l=-\big[\lambda^{(i)}_{k-i}-\lambda^{(i)}_{k-i+1}\big]-\big[\lambda^{(i)}_{k-i-\ell N}-\lambda^{(i)}_{k-i}\big] = -\theta-m
\end{gather*}
with $m=\lambda^{(i)}_{k-i-\ell N}-\lambda^{(i)}_{k-i}\geq 0$.
This proves that~\eqref{frac2} holds true if
\begin{gather}\label{C1C2C}
C_2\geq 1+|1-q/t|\frac{1}{|1-|q||}
\end{gather}
for all cases $j-i\in N\Z_{\geq 0}$.

We~proved that~\eqref{frac1}--\eqref{frac2} holds true for {\em all} cases provided the conditions in~\eqref{C1C2A},~\eqref{C1C2B} and~\eqref{C1C2C} all hold true; this is the case if we choose $C_1$ and $C_2$ as in~\eqref{C1C2}.

\subsection[Proof of~Lemma~B1]{Proof of~Lemma~\ref{lem:est}}\label{app:prooflem:est}
\subsubsection[Proof of~the estimate in (B.5)]{Proof of~the estimate in~(\ref{esta})}
We~have
\begin{gather*}
\text{LHS} = \left|\prod_{n=0}^{\theta-1}\frac{1-q^{l+n}au}{1-q^{l+n}u} \right| = \prod_{n=0}^{\theta-1}\left|1+(1-a)\frac{q^{l+n}u}{1-q^{l +n}u} \right|
\\ \hphantom{\text{LHS}}
{}\leq \prod_{n=0}^{\theta-1}\left(1+|1-a|\left|\frac{q^{l+n}u}{1-q^{l+n}u}\right| \right)\leq \prod_{n=0}^{\theta-1}\left(1+\frac{|1-a|}{|\sin\arg(u)|}\right) =
\text{RHS},
\end{gather*}
using the estimate in~\eqref{zest} and $\big|\sin\arg\big(q^{l+n}u\big)\big|=|\sin\arg(u)|$ since $q$ is real.

\subsubsection[Proof of~the estimate in (B.6)]{Proof of~the estimate in~(\ref{estb})}
\label{app:estb}
We~have
\begin{gather*}
\text{LHS} = \left| \prod_{n=0}^{\theta-1} \frac{1- q^{n-\theta-m+1}a\kappa^\ell }{1-q^{n-\theta-m+1} \kappa^\ell } \right| = \prod_{n=0}^{\theta-1}\left| 1+(1-a)\frac{q^{-n-m}\kappa^\ell}{1-q^{-n-m}\kappa^\ell}\right|
\\ \hphantom{\text{LHS}}
{}\leq \prod_{n=0}^{\theta-1}\left( 1+|1-a|\left|\frac{q^{-n-m}\kappa^\ell}{1-q^{-n-m}\kappa^\ell}\right|\right)
\leq \prod_{n=0}^{\theta-1}\left( 1+|1-a|\frac{|\kappa|}{|1-|\kappa||}\right) =\text{RHS},
\end{gather*}
using
\begin{gather*}
\left|\frac{q^{-l}\kappa^\ell}{1-q^{-l}\kappa^\ell}\right|\leq \frac{|\kappa|}{|1-|\kappa||}\qquad
(l\geq 0, \ell\geq 1);
\end{gather*}
the latter follows for the case $|\kappa|<1$ and $|q|>1$ from the following inequality: $x/(1-x)<y/(1-x)$ for~$0\leq x<y<1$, and for the case $|\kappa|>1$ and $|q|<1$:
\begin{gather*}
\left|\frac{q^{-l}\kappa^\ell}{1-q^{-l}\kappa^\ell}\right| = \left|\frac{1}{1-q^{l}\kappa^{-\ell}}\right| \leq \frac1{1-|1/\kappa|} = \frac{|\kappa|}{|1-|\kappa||}
\end{gather*}
since $1/(1-x)<1/(1-y)$ for~$0\leq x<y<1$.

\subsubsection[Proof of~the estimate in (B.7)]{Proof of~the estimate in~(\ref{estc})}
We~have
\begin{gather*}
\text{LHS} = \left| \prod_{n=0}^{\theta-1} \frac{1- q^{n-\theta-m}a\kappa^\ell }{1-q^{n-\theta-m+1} \kappa^\ell } \right| = \prod_{n=1}^{\theta}\left| 1+(1-a)\frac{q^{-n-m}\kappa^\ell}{1-q^{-n-m}\kappa^\ell}\right|
\\ \hphantom{\text{LHS}}
{}\leq \prod_{n=1}^{\theta}\left( 1+|1-a|\left|\frac{q^{-n-m}\kappa^\ell}{1-q^{-n-m}\kappa^\ell}\right|\right)
\leq \prod_{n=1}^{\theta}\left( 1+|1-a|\frac{1}{|1-|q||}\right) =\text{RHS},
\end{gather*}
using
\begin{gather*}
\left|\frac{q^{-l}\kappa^\ell}{1-q^{-l}\kappa^\ell}\right|\leq \frac{|q^{-1}|}{||q^{-1}|-1|} = \frac{1}{|1-|q||} \qquad
(l\geq 1,\ \ell\geq 0),
\end{gather*}
as in the proof of~\eqref{estb}.

\section[Proof of Proposition~4.2]{Proof of Proposition~\ref{prop:cT}}\label{app:cT}
We~prove Proposition~\ref{prop:cT} using Lemma~\ref{lem:cTtrig} in the main text.

\subsection{Proof of~commutativity}
We~prove~\eqref{commutativity}. We~note that the action of~the Macdonald--Ruijsenaars operators in~\eqref{MR} on functions $x^\lambda f(x)$ can be written as \cite{NS}
\begin{gather}\label{cD}
D_N^\pm(x|q,t)x^\lambda f(x) = x^\lambda E^\pm_N(x|s|q,t)f(x) ,\qquad
s_i=t^{N-i}q^{\lambda_i}
\end{gather}
with the modified Macdonald--Ruijsenaars operators
\begin{gather}
\label{Edef}
E^\pm_N(x|s|q,t)= \sum_{i=1}^N A_{N,i}\big(x|t^{\pm }\big)s_i^{\pm 1}T_{q,x_i}^{\pm 1} ,
\\
\label{Adef}
A_{N,i}(x|t^{\pm 1}) = \prod_{j=1}^{i-1}\frac{1-t^{\pm1}x_i/x_j}{1-x_i/x_j} \prod_{k=i+1}^N\frac{1-t^{\mp 1}x_k/x_i}{1-x_k/x_i}
\end{gather}
(this can be proved by simple computations which we skip).

We~also need properties of~the function $\chi_N(x|y|q,t)$ in~\eqref{psiN} which we summarize as follows.
\begin{Lemma}
\label{lem:psiN}
The function $\chi_N(x|y|q,t)$ satisfies the following duality relation,
\begin{gather}
\label{duality}
\chi_N(x|y|q,t)=\chi_N(y|x|q,t/q).
\end{gather}
Moreover,
\begin{gather}
\begin{split}
\label{cDpsi}
&E^\pm_N(x|y|q,t) \chi_N(x|y|q,t)=e_1\big(y^{\pm 1}\big) \chi_N(x|y|q,t),
\\[.5ex]
&E^\pm_N(y|x|q,q/t) \chi_N(x|y|q,t)=e_1\big(x^{\pm 1}\big) \chi_N(x|y|q,t)
\end{split}
\end{gather}
with
\begin{gather}
\label{e1}
e_1\big(x^{\pm 1}\big) = x_1^{\pm 1}+\cdots + x_N^{\pm 1}.
\end{gather}
\end{Lemma}
\begin{proof}The definitions in~\eqref{phiN} and~\eqref{psiN} imply
\begin{gather*}
\chi_N(x|y|q,t) = \prod_{1\leq i<j\leq N} \frac{(qx_j/x_i;q)_\infty}{(qx_j/tx_i;q)_\infty} \frac{(qy_j/y_i;q)_\infty}{(ty_j/y_i;q)_\infty} \, \varphi_N(x|y|q,t) .
\end{gather*}
The product on the right-hand side is manifestly invariant under the transformation $(x,y,t)\mapsto(y,x,q/t)$;
the function $\varphi_N(x|y|q,t)$ has this invariance by~\eqref{symmetry}. This proves~\eqref{duality}.

The first identity in~\eqref{cDpsi} is implied by $E^\pm_N(x|s|q,t)f_N(x|s|q,t)=e_1\big(s^{\pm 1}\big)f_N(x|s|q,t)$ proved in \cite{NS}; the second follows from the first and the duality in~\eqref{duality}.
\end{proof}

Equation~\eqref{cD} and Lemma~\ref{lem:cTtrig} imply that the result we want to prove:
\begin{gather*}
D_N^\pm(x|q,t)\cT_N(x|q,t)x^\lambda f(x)=\cT_N(x|q,t) D_N^\pm(x|q,t)x^\lambda f(x),
\end{gather*}
is equivalent to
\begin{gather*}
\Bigg[ E^\pm_N(x|s|q,t) \prod_{i=1}^N\vartheta_3(s_ix_i/y_i|q)\chi_N(x|y|q,t) \prod_{1\leq i<j\leq N}(1-y_j/y_i)f(y)\Bigg]_{1,y}
\\ \qquad
{}= \Bigg[\prod_{i=1}^N\vartheta_3(s_ix_i/y_i|q)\chi_N(x|y|q,t) \prod_{1\leq i<j\leq N}(1-y_j/y_i)\big( E^\pm_N(y|s|q,t) f(y)\big) \Bigg]_{1,y}\! .
\end{gather*}
The latter is obviously implied by the following two identities: first,
\begin{gather}
E^\pm_N(x|s|q,t) \prod_{i=1}^N\vartheta_3(s_ix_i/y_i|q)\chi_N(x|y|q,t) \nonumber
\\ \qquad
{}=E^\mp_N\big(y\big|s^{-1}\big|q,q/t\big) \prod_{i=1}^N\vartheta_3(s_ix_i/y_i|q)\chi_N(x|y|q,t) ,
\label{toprove1}
\end{gather}
and second,
\begin{gather}
\Bigg[ \Bigg( E^\mp_N\big(y\big|s^{-1}\big|q,q/t\big) \prod_{i=1}^N\vartheta_3(s_ix_i/y_i|q)\chi_N(x|y|q,t)\Bigg) \prod_{1\leq i<j\leq N}(1-y_j/y_i)f(y)\Bigg]_{1,y} \nonumber
\\ \qquad
{}= \Bigg[\prod_{i=1}^N\vartheta_3(s_ix_i/y_i|q)\chi_N(x|y|q,t) \prod_{1\leq i<j\leq N}(1-y_j/y_i)\big( E^\pm_N(y|s|q,t) f(y)\big) \Bigg]_{1,y}\! .
\label{toprove2}
\end{gather}

We~first prove~\eqref{toprove1} in three steps, using the shorthand notation in~\eqref{e1}. We~start with
\begin{gather}
\label{Id1}
E^\pm_N(x|s|q,t) \prod_{i=1}^N\vartheta_3(s_ix_i/y_i|q) = \frac{q^{-1/2}}{1-q^{-1}}\prod_{i=1}^N\vartheta_3(s_ix_i/y_i|q) \big[e_1\big(x^{\mp 1}\big),E^\pm_N(x|y|q,t)\big]
\end{gather}
proved by the following computation (we insert definitions and change the summation variable $n_i\pm 1\to n_i$ in the third equality),
\begin{gather*}
\text{LHS} = \sum_{i=1}^N A_{N,i}\big(x|t^{\pm 1}\big) s^{\pm 1}_i T_{q,x_i}^{\pm 1} \sum_{n\in\Z^N} \prod_{j=1}^N\left(\frac{s_j x_j}{y_j} \right)^{n_j}q^{\frac12 n_j^2}
\\ \hphantom{\text{LHS}}
{}= \sum_{n\in\Z^N} \sum_{i=1}^N A_{N,i}\big(x|t^{\pm 1}\big) \left(\frac{y_i}{x_i}\right)^{\pm 1} \left(\frac{s_i x_i}{y_i} \right)^{n_i\pm 1}q^{\frac12n_i^2\pm n_i} \prod_{\setatop{j=1}{j\neq i}}^N\left(\frac{s_j x_j}{y_j} \right)^{n_j}q^{\frac12 n_j^2} T_{q,x_i}^{\pm 1}
\\ \hphantom{\text{LHS}}
{}= \prod_{i=1}^N\vartheta_3(s_ix_i/y_i|q) \sum_{i=1}^N A_{N,i}\big(x|t^{\pm 1}\big) \left( \frac{y_i}{x_i}\right)^{\pm 1} q^{-1/2} T_{q,x_i}^{\pm 1} =\text{RHS}
\end{gather*}
since
\begin{gather*}
\big[e_1\big(x^{\mp 1}\big),E^\pm_N(x|y|q,t)\big]
\\
\qquad
{}= \sum_{i=1}^N A_{N,i}\big(x|t^{\pm 1}\big) y_i^{\pm 1}\big[x_i^{\mp 1},T_{q,x_i}^{\pm 1}\big]
 = \big(1-q^{-1}\big)\sum_{i=1}^N A_{N,i}\big(x|t^{\pm 1}\big) \left(\frac{y_i}{x_i}\right)^{\pm 1}T_{q,x_i}^{\pm 1} .
\end{gather*}
Next,
\begin{gather}
\label{Id2}
\big[e_1\big(x^{\mp 1}\big),E^\pm_N(x|y|q,t)\big] \chi_N(x|y|q,t)= \big[e_1\big(y^{\pm 1}\big),E^\mp_N(y|x|q,q/t)\big] \chi_N(x|y|q,t) ,
\end{gather}
which is proved by
\begin{gather*}
\text{LHS} = \left(e_1\big(x^{\mp 1}\big)e_1\big(y^{\pm 1}\big)- E^\pm_N(x|y|q,t)E^\mp_N(y|x|q,q/t) \right)\chi_N(x|y|q,t)
\\ \hphantom{\text{LHS}}
{}= \left(e_1\big(y^{\pm 1}\big) e_1\big(x^{\mp 1}\big)- E^\mp_N(y|x|q,q/t)E^\pm_N(x|y|q,t) \right)\chi_N(x|y|q,t) = \text{RHS}
\end{gather*}
using~\eqref{cDpsi} and
\begin{gather*}
\big[ E^\pm_N(x|y|q,t),E^\mp_N(y|x|q,q/t)\big]=0;
\end{gather*}
the latter is verified by a~simple computation using the definition in~\eqref{Edef}--\eqref{Adef}. Third,
\begin{gather}
E^\mp_N\big(y|s^{-1}|q,q/t\big) \prod_{i=1}^N\!\vartheta_3(s_ix_i/y_i|q) \nonumber\\
\qquad{} = \frac{q^{-1/2}}{1-q^{-1}} \prod_{i=1}^N\!\vartheta_3(s_ix_i/y_i|q) \big[e_1\big(y^{\pm 1}\big),E^\mp_N(y|x|q,q/t)\big],\label{Id3}
\end{gather}
which is proved similarly as~\eqref{Id1}:
\begin{gather*}
\text{LHS} = \sum_{i=1}^N A_{N,i}\big(y|(q/t)^{\mp 1}\big) s^{\pm 1}_i T_{q,y_i}^{\mp 1} \sum_{n\in\Z^N} \prod_{j=1}^N\left(\frac{s_j x_j}{y_j} \right)^{n_j}q^{\frac12 n_j^2}
\\ \hphantom{\text{LHS}}
{}= \sum_{n\in\Z^N} \sum_{i=1}^N A_{N,i}\big(y|(q/t)^{\mp 1}\big) \left(\frac{y_i}{x_i}\right)^{\pm 1} \left(\frac{s_i x_i}{y_i} \right)^{n_i\pm 1}q^{\frac12n_i^2\pm n_i} \prod_{\setatop{j=1}{j\neq i}}^N\left(\frac{s_j x_j}{y_j} \right)^{n_j}q^{\frac12 n_j^2} T_{q,x_i}^{\mp 1}
\\ \hphantom{\text{LHS}}
{}= \prod_{i=1}^N\vartheta_3(s_ix_i/y_i|q) \sum_{i=1}^N A_{N,i}\big(y|(q/t)^{\mp 1}\big) \left( \frac{y_i}{x_i}\right)^{\pm 1} q^{-1/2} T_{q,y_i}^{\mp 1} =\text{RHS}
\end{gather*}
since
\begin{gather*}
\big [e_1\big(y^{\pm 1}\big),E^\mp_N(y|x|q,q/t)\big]
 \\ \qquad{}= \sum_{i=1}^N A_{N,i}\big(y|(q/t)^{\mp 1}\big) x_i^{\mp 1}\big[y_i^{\pm 1},T_{q,y_i}^{\mp 1}\big]
 = \big(1-q^{-1}\big)\sum_{i=1}^N A_{N,i}\big(y|(q/t)^{\mp 1}\big) \left(\frac{y_i}{x_i}\right)^{\pm 1}T_{q,y_i}^{\mp 1}.
\end{gather*}
We~are now ready to prove~\eqref{toprove1}: we insert~\eqref{Id1} into the LHS in~\eqref{toprove1}, use~\eqref{Id2} and~\eqref{Id3}, and obtain the RHS in~\eqref{toprove1}.

To conclude our proof of~\eqref{commutativity}, we prove~\eqref{toprove2} by the following computation, using the definitions in~\eqref{Edef}--\eqref{Adef} and the basic property $\left[ \left(T_{q,y_i}^{\mp 1}g_1(y)\right) g_2(y)\right]_{1,y} = \left[ g_1(y)\left(T_{q,y_i}^{\pm 1}g_2(y)\right)\right]_{1,y}$ of~the constant term:

\begin{gather*}
\text{LHS} \,{=}\, \Biggl[ \sum_{i=1}^N\! A_{N,i}\big(y|(q/t)^{\mp 1}\big)s_i^{\pm 1} \!
\Bigg(\! T_{q,y_i}^{\mp1} \prod_{i=1}^N\!\vartheta_3(s_ix_i/y_i|q)\chi_N(y|x|q,t)\!\Bigg)\!\!\!
 \prod_{1\leq i<j\leq N}\!\!\!\!\!(1-y_j/y_i)f(y)\!\Biggr]_{1,y}
 \\ \hphantom{\text{LHS}}
{}\,{=}\, \Biggl[\prod_{i=1}^N\vartheta_3(s_ix_i/y_i|q)\chi_N(y|x|q,t)
\!\sum_{i=1}^N\! \Bigg(\! T_{q,y_i}^{\pm 1} A_{N,i}\big(y|(q/t)^{\mp 1}\big)s_i^{\pm 1}\!\!\! \prod_{1\leq i<j\leq N}\!\!\!\!\!(1-y_j/y_i)f(y)\!\Bigg) \!\Biggr]_{1,y}
 \\ \hphantom{\text{LHS}}
{}=\text{RHS}
\end{gather*}
provided
\begin{gather*}
\frac{ T_{q,y_i}^{\pm 1}A_{N,i}\big(y|(q/t)^{\mp 1}\big) \prod_{1\leq i<j\leq N}(1-y_j/y_i)}{\prod_{1\leq i<j\leq N}(1-y_j/y_i)} = A_{N,i}\big(y|t^\pm\big) ;
\end{gather*}
the latter holds true for the coefficients $A_{N,i}(x|t)$ in~\eqref{Adef}, as is easily verified:
\begin{gather*}
\text{LHS} =
\prod_{j=1}^{i-1}\frac{1-(q/t)^{\mp 1}q^{\pm 1}y_i/y_j}{1-q^{\pm 1}y_i/y_j} \prod_{k=i+1}^N\frac{1-(q/t)^{\pm 1}y_k/q^{\pm 1}y_i}{1-y_k/q^{\pm 1}y_i}
 \prod_{j=1}^{i-1}\frac{1-q^{\pm 1}y_i/y_j}{1-y_i/y_j}
 \\ \hphantom{\text{LHS}=}
 {}\times\prod_{k=i+1}^N \frac{1-y_k/q^{\pm 1}y_i}{1-y_k/y_i}
=\text{RHS}.
\end{gather*}

\subsection{Eigenfunction property}
The eigenfunctions $x^\lambda f_N(x|s|q,t)$ of~$D^\pm_N(x|s|q,t)$ are unique, and~\eqref{commutativity} therefore implies that $x^\lambda f_N(x|s|q,t)$ also are eigenfunction of~$\cT_N(x|q,t)$.

We~are left to determine the eigenvalues. For that, we introduce some notation: the space of~formal power series $\C[[x_2/x_1,\dots,x_N/x_{N-1}]]$ is spanned by monomials
\begin{gather*}
(x_2/x_1)^{\alpha_1} \cdots (x_N/x_{N-1})^{\alpha_{N_1}}
\end{gather*}
with $\alpha=(\alpha_1,\dots,\alpha_{N-1})\in\Z_{\geq 0}^{N-1}$.
Any such monomial can be written as $x^\mu=x_1^{\mu_1}\cdots x_N^{\mu_N}$ with $\mu=\mu(\alpha)$ given by
\begin{gather*}
\mu_1(\alpha)=-\alpha_1,\qquad \mu_i(\alpha)=\alpha_{i-1}-\alpha_i\qquad (i=2,\dots,N-1),\qquad \mu_N(\alpha)=\alpha_{N-1}.
\end{gather*}
One can verify that the action of~the operator in~\eqref{cTtrig} is triangular on this basis in the following sense,
\begin{gather*}
\cT_N(x|q,t) x^\lambda x^{\mu(\alpha)} = \varepsilon(\lambda+\mu(\alpha))x^\lambda \Bigg( x^{\mu(\alpha)} + \sum_{\beta>\alpha} v_{\alpha\beta} x^{\mu(\beta)} \Bigg)
\end{gather*}
for some coefficients $v_{\alpha\beta}$, where $\beta\geq \alpha$ means that $\beta_i\geq \alpha_i$ for all $i=1,\dots,N-1$ (this follows from
\begin{gather*}
q^{\frac12\Delta} x^\lambda x^{\mu} = \varepsilon(\lambda+\mu)x^{\lambda +\mu} ,\qquad \varepsilon(\lambda+\mu)=q^{\frac12\sum_{i=1}^N(\lambda_i+\mu_i+(N-i)\beta)^2}
\end{gather*}
used already in the main text, and the fact that all functions of~$x$ appearing in the definition of~$\cT_N(x|q,t)$ in~\eqref{cTtrig} can be expanded as power series in $\C[[x_2/x_1,\dots,x_N/x_{N-1}]]$).
Since
\begin{gather*}
x^\lambda f_N(x|s|q,t) = x^\lambda\Bigg( 1 + \sum_{\alpha>0} b_\alpha x^{\mu(\alpha)}\Bigg)
\end{gather*}
for some coefficients $b_\alpha$, the eigenvalue is $\varepsilon(\lambda)=\varepsilon(s|q)$ in~\eqref{vepsN11}.

\section{Identities}\label{app:Id}
For clarify, and for the convenience of~the reader, we state and prove two identities used in the main text.

First, the identity
\begin{gather}\label{Euler}
\sum_{\lambda\in\mathsf{P}}\alpha^{|\lambda|} = \frac1{(\alpha;\alpha)_\infty} \qquad (|\alpha|<1),
\end{gather}
which goes back to Euler, is important in our proof of~Theorem~\ref{thm:converegence}. It is proved by the following elementary computation making absolute convergence of~the series manifest,
\begin{gather}
\text{LHS} = \lim_{M\to\infty} \sum_{\lambda_1\geq \lambda_2\geq \cdots\geq \lambda_M\geq 0} \alpha^{\lambda_1+\lambda_2+\dots+\lambda_M}\nonumber
\\ \phantom{\text{LHS}}
{}= \lim_{M\to\infty} \sum_{\lambda_1=\lambda_2}^\infty \sum_{\lambda_2=\lambda_3}^\infty \cdots \sum_{\lambda_{M-1}=\lambda_M}^\infty \sum_{\lambda_M=0}^\infty \alpha^{\lambda_1+\lambda_2+\dots+\lambda_M}\nonumber
\\ \phantom{\text{LHS}}
{}= \lim_{M\to\infty} \sum_{\lambda_2=\lambda_3}^\infty \cdots \sum_{\lambda_{M-1}=\lambda_M}^\infty \sum_{\lambda_M=0}^\infty \frac{1}{1-\alpha}
\alpha^{2\lambda_2+\lambda_3 +\dots+\lambda_M}\nonumber
\\ \phantom{\text{LHS}}
{}= \lim_{M\to\infty} \sum_{\lambda_3=\lambda_4}^\infty \cdots \sum_{\lambda_{M-1}=\lambda_M}^\infty \sum_{\lambda_M=0}^\infty \frac{1}{(1-\alpha)(1-\alpha^2)}
\alpha^{3\lambda_3+\lambda_4 +\dots+\lambda_M} = \cdots \nonumber
\\ \phantom{\text{LHS}}
{}= \lim_{M\to\infty} \frac1{(1-\alpha)} \frac1{(1-\alpha^2)}\cdots \frac1{(1-\alpha^M)} = \text{RHS},
\label{Euler1}
\end{gather}
summing repeatedly the geometric series.

Second, we state and prove an identity used in the proof of~Fact~\ref{fact:Id} and Remark~\ref{rem:Id}.

\begin{Lemma}\label{lem:Id2}For $p\in\C$, let $x_i\in\C$ be given for~$i=1,\dots,N$, and extend this definition to all $i\geq 1$ by $x_{i+N}=px_i$. Then
\begin{gather*}
\prod_{i=1}^N \prod_{j>i} \frac{(ax_j/x_i;q)_\infty}{(bx_j/x_i;q)_\infty} =
\prod_{1\leq i<j\leq N}\frac{(ax_j/x_i;q,p)_\infty}{(bx_j/x_i;q,p)_\infty}
\prod_{1\leq i\leq j\leq N}\frac{(pax_i/x_j;q,p)_\infty}{(pbx_i/x_j;q,p)_\infty}
\end{gather*}
for all $a, b, q\in\C$.
\end{Lemma}

\begin{proof}We~note that $(z;q,p)= \prod_{m=0}^\infty (p^m x;q)_\infty$, and thus
\begin{gather*}
\text{LHS} =
 \prod_{1\leq i<j\leq N} \prod_{m=0}^\infty \frac{(ax_{j+m N}/x_i;q)_\infty}{(bx_{j+m N}/x_i;q)_\infty} \prod_{1\leq j\leq i\leq N} \prod_{m=1}^\infty \frac{(ax_{j+m N}/x_i;q)_\infty}{(bx_{j+m N}/x_i;q)_\infty}
 \\ \phantom{\text{LHS}}
{}= \prod_{1\leq i<j\leq N} \prod_{m=0}^\infty \frac{(p^m ax_{j}/x_i;q)_\infty}{(p^m b x_{j}/x_i;q)_\infty} \prod_{1\leq j\leq i\leq N}\prod_{m=1}^\infty \frac{(p^m a x_{j}/x_i;q)_\infty}{(p^m bx_{j}/x_i;q)_\infty}
= \text{RHS},
\end{gather*}
inserting $x_{j+m N}=p^mx_j$.
\end{proof}

\subsection*{Acknowledgements}

We~would like to thank F.~Atai, A.~Negut, and V.~Pasquier for useful discussions.
We~thank J.~Lamers for a~helpful comment on the manuscript.
We~are grateful to careful referees for remarks helping us to improve our paper.
This work is supported by VR Grant No~2016-05167 (E.L.) and by JSPS Kakenhi Grants (B) 15H03626 (M.N.), (C) 19K03512 (J.S.).
MN gratefully acknowledges financial support by the Knut and Alice Wallenberg foundation (KAW 2019.0525).
We~are grateful to the Stiftelse Olle Engkvist Byggm\"astare, Contract 184-0573, for~financial support.

\pdfbookmark[1]{References}{ref}
\LastPageEnding

\end{document}